\RequirePackage{etex}
\documentclass[a4paper,12pt,reqno,twoside, flalign]{article}
\usepackage[margin=2.0cm]{geometry}
\usepackage{setspace}
\usepackage[normalem]{ulem}
\usepackage{amsmath, mathtools, calligra}

\usepackage{amsfonts}
\usepackage{amssymb}
\usepackage[svgnames, x11names]{xcolor}
\usepackage{graphicx}
\usepackage{array}
\usepackage{pictexwd}
\usepackage{fancybox}
\usepackage{natbib}
\usepackage{newcent}
\usepackage{enumitem}
\usepackage{tikz}
\usepackage{mathrsfs}
\usepackage{bm}
\usepackage{subcaption}
\usepackage[section]{placeins}
\setlist{nolistsep}
\newcolumntype{C}[1]{>{\centering\let\newline\\\arraybackslash\hspace{0pt}}m{#1}}
\usepackage{diagbox}
\usepackage{tikz}
\usetikzlibrary{arrows.meta, positioning}

\DeclareMathAlphabet{\mathpzc}{OT1}{pzc}{m}{it}

\newtheorem{theorem}{Theorem}
\newtheorem{proposition}{Proposition}

\newtheorem{corollary}{Corollary}
\newtheorem{lemma}{Lemma}

\newtheorem{definition}{Definition}
\newtheorem{example}{Example}
\newtheorem{remark}{Remark}
\newenvironment{proof}[1][]
      {\par\medbreak{\noindent\bfseries Proof#1\quad}}
      {\hfill $\blacksquare$\bigbreak}
\def \co{ {\rm co\,} }

\def \beq {\begin{eqnarray*}}
\def \eeq {\end{eqnarray*}}

\usepackage{comment}

\usepackage[textwidth=30mm]{todonotes}

\usepackage[colorlinks=true, citecolor=  DodgerBlue4, linkcolor=DodgerBlue4, urlcolor=black]{hyperref}

\usepackage{mathtools}

\begin{document}
	\title{Non-Bayesian Learning in Misspecified Models\thanks{We are grateful to Michel Benaim and Peter Gr\"unwald for their invaluable help and to  Aislinn Bohren, Daniel Hauser, Jan Knoepfle, Pietro Ortoleva, Mallesh Pai, Maher Said, Larry Samuelson, Ran Spiegler, Jakub Steiner, Philipp Strack, and seminar audiences for insightful discussions. 
	LR thanks the generous hospitality of AMSE and LUISS. SB and MF acknowledge financial support from the French government under the “France 2030” investment plan managed by the French National Research Agency Grant ANR-17-EURE-0020, and by the Excellence Initiative of Aix-Marseille University - A*MIDEX.}}

	\date{Latest version: \today \\ First version: June 26, 2024}
%	\thanks{ }
	\author{Sebastian Bervoets\thanks{Aix-Marseille Univ., CNRS, AMSE, sebastian.bervoets(at)univ-amu.fr} \;\& Mathieu Faure\thanks{Aix-Marseille Univ., CNRS, AMSE, mathieu.faure(at)univ-amu.fr} \;\& Ludovic Renou\thanks{ASU, QMUL and  CEPR, lrenou.econ(at)gmail.com }}
\maketitle

	\begin{abstract} Deviations from Bayesian updating are traditionally categorized as biases, errors, or fallacies, thus implying their inherent ``sub-optimality.'' We offer a more nuanced view. In learning problems with misspecified models, we show that some non-Bayesian updating can outperform Bayesian updating.


	\end{abstract}

\newpage 
\section{Introduction}\label{sec:intro}
Psychologists, economists, and other social scientists have extensively documented systematic departures from Bayesian updating. These departures include underreaction and overreaction to new information, confirmation bias, base-rate neglect, and numerous other phenomena.\footnote{For comprehensive reviews, see \cite{benjamin2019} and \cite{ortoleva2022}.} While these departures are traditionally characterized as biases, errors, or fallacies -- implying their inherent ``sub-optimality''  --  we offer a more nuanced view. Specifically, we demonstrate that in classical learning problems, some non-Bayesian updating can outperform its Bayesian counterpart.

\medskip 

The benefits of non-Bayesian updating have precedents in the literature, particularly in strategic environments where such departures can induce equilibrium behavior favorable to non-Bayesian agents. In sender-receiver games, for instance, non-Bayesian updating by receivers can elicit more information from the  senders \citep{declippel-zhang-2022, lee2023}.\footnote{As another instance, in a financial trading problem, \cite{massari-22} shows that non-Bayesian traders may drive Bayesian traders out of the market. As our analysis will make it clear, this is because non-Bayesian traders may better predict the underlying fundamentals than their Bayesian counterparts.} However, these results rely fundamentally on strategic interactions. In single-agent learning environments, where strategic responses are absent, we identify a distinct mechanism through which non-Bayesian updating may confer a benefit, which we now explain.

\medskip

We formalize our analysis within the canonical learning framework of \cite{Ber66}. Consider an agent attempting to learn a true data generating process. The agent begins with a set of candidate processes, none of which may exactly match the true process. Then, in each of an infinite number of periods, the agent sequentially observes a new piece of data and updates their beliefs about which process best describes the observations. We model the agent as a \emph{conservative} Bayesian, following the updating rule:
\begin{align*}\tag{\ref{C-Bay}}
(\text{belief at } n+1) = (1-\gamma) \times (\text{belief at } n) + \gamma \times (\text{Bayesian update at } n+1),
\end{align*}
where $\gamma$ may depend on the number of observations, current belief, and current observation. This rule, first introduced by \cite{phillips1966} and \cite{edwards1968}, generalizes Bayesian updating and captures several documented cognitive biases, including under-reaction and over-reaction to new information,  confirmation bias, and base-rate neglect. Throughout, we refer to this updating rule as the (\ref{C-Bay}) rule.\medskip

In  the first part of the analysis, we  assume that the more observations the agent has, the less the agent reacts to new information. A body of evidence supports this assumption; in his meta-analysis, \cite{benjamin2019} reports it as Stylized Fact 2. Our main result -- Theorem \ref{th:limitset} -- states that the agent's \emph{predictive} process -- the expectation of candidate processes under the agent's beliefs -- converges almost surely to the \emph{mixture} process the closest to the true data generating process (in the sense of the Kullback-Leibler divergence). In other words, the theorem states that it is \emph{as if} the agent considers all possible mixtures of the candidate processes (thus, ``convexifying'' his model) and learns as a classical Bayesian. In particular, if the true data generating process is a mixture of the candidate processes,  the agent's predictive process converges to it. This contrasts with the Bayesian agent, whose predictive process converges to the \emph{pure} process the closest to the true data generating process -- see \cite{Ber66}.  (We stress that when the model is correctly specified, the updating rule (\ref{C-Bay}) converges to the true data-generating process, as Bayesian updating does.) \medskip

The intuition is simple. Bayesian agents extract all the information there is from their models.\footnote{\cite{zellner1988} proves that Bayesian updating corresponds to extracting all the information from the prior and the statistical model.} However, when their models are wrong, this typically causes them to make wrong inferences. A more conservative approach, as embodied by the updating rule (\ref{C-Bay}), limits the severity of this problem. Moreover, the more observations an agent has, the more conservative they must be to counteract the increasing confidence of their Bayesian counterpart. Ultimately, this can lead to better learning outcomes. 

\medskip

Mathematically, we rewrite the updating rule (\ref{C-Bay}) as the Robbins-Monro algorithm \citep{robbins-monro-51}: 
\begin{align*}
(\text{belief at\;} n+1) & = (\text{\;belief at\;} n) + \gamma_{n+1} \Bigl[\mathbb{E}_{\text{\;true DGP}}( \text{Bayesian update at } n+1) - (\text{belief at\;}  n)\Bigr]\\
 & + \gamma_{n+1} \Bigl[( \text{Bayesian update at } n+1)  -\mathbb{E}_{\text{\;true DGP}}( \text{\;Bayesian update at } n+1)\Bigr],
 %\tag{\ref{C-Bay}}
\end{align*}
where $\mathbb{E}_{\text{\;true DGP}}$ is the expectation with respect to the true data-generating process.  The first term in brackets is deterministic, while the second term  is random (a martingale difference). When the agent reacts sufficiently slowly to new information, the random term becomes negligible as the number of observations grows arbitrarily large. The deterministic part then drives the belief dynamics. We show that the cross-entropy of the predictive process with respect to the true data generating process acts as a potential for the deterministic dynamics.\footnote{The KL divergence of $p$ from $q$ is $\mathbb{E}_{p}[\log(p/q)]$. Minimizing the KL divergence with respect to $q$ is equivalent to maximizing $\mathbb{E}_p[\log q]$. We call the latter the cross-entropy of $q$ with respect to $p$.}   We then show that all critical points of the deterministic flow correspond to either the unconstrained maximizers of the cross-entropy or  constrained ones. Finally, we prove that all constrained maximizers repel the dynamics, hence the dynamics converges to the unconstrained maximizers. This last step is the most delicate. The constrained maximizers are neither isolated points nor interior, which precludes us from appealing  to classical results in stochastic approximation theory such as \cite{pemantle1990},  \cite{brandiere1998}, or \cite{brandiere-duflo-1996}. We develop new arguments to solve that issue. \medskip

In the second part of the analysis, we consider alternative assumptions on the updating weights and show  that our main insight, that is, that some non-Bayesian updating can outperform Bayesian updating, continues to hold. For instance, if the updating weights are constant, then the belief dynamics does not converge, but its occupation measure does. As the weight $\gamma$ becomes arbitrarily small, the occupation measure concentrates on the mixture process the closest to the true data generating process -- Theorem \ref{th:constant-step}. \medskip

We illustrate our main result with a simple example, inspired by the work of \cite{spiegler2016} on causal models. Suppose that each observation is a vector of variables $(x_1,x_2,x_3)$, representing diet, sleep, and health, respectively. The true causal model is such that diet and sleep function are \emph{interdependent} determinants of health, with diet also influencing sleep patterns. However, the agent posits another causal model and assumes that diet and sleep are \emph{independent} causes of health. Figure \ref{fig:intro-DAG} represents the causal models as two directed acyclic graphs. 
Thus, the agent considers distributions of the form  $p(x^3|x^1,x^2)p(x^2)p(x^1)$, which all differ from the true distribution since $x^1$ and $x^2$ are correlated in the true model. Since any correlated distribution is the convex combination of independent distributions, our main result implies that the non-Bayesian agent successfully converges to the correct causal model, while the Bayesian agent fails to do so. This is the paper's main insight.\medskip

\begin{figure}[h!]
    \centering
    
    % Define the style for nodes and arrows
    \tikzset{
        node style/.style={circle, draw, minimum size=1cm},
        arrow style/.style={-{Stealth[length=3mm]}, thick}
    }
    
    \begin{tikzpicture}[node distance=2cm]
        % Left panel - True causal model
        \begin{scope}[local bounding box=left]
            % Nodes
            \node[node style] (x3) at (0,4) {$x_3$};
            \node[node style] (x1) at (-1.5,0) {$x_1$};
            \node[node style] (x2) at (1.5,0) {$x_2$};
            
            % Arrows - all straight
            \draw[arrow style] (x1) -- (x2);
            \draw[arrow style] (x1) -- (x3);
            \draw[arrow style] (x2) -- (x3);
            
            % Caption
            \node[below=0.5cm] at (0,-1) {The true causal model};
        \end{scope}
        
        % Right panel - Agent's model
        \begin{scope}[xshift=7cm, local bounding box=right]
            % Nodes
            \node[node style] (y3) at (0,4) {$x_3$};
            \node[node style] (y1) at (-1.5,0) {$x_1$};
            \node[node style] (y2) at (1.5,0) {$x_2$};
            
            % Arrows - all straight
            \draw[arrow style] (y1) -- (y3);
            \draw[arrow style] (y2) -- (y3);
            
            % Caption
            \node[below=0.5cm] at (0,-1) {The agent's causal model};
        \end{scope}
    \end{tikzpicture}
    \caption{The causal models} \label{fig:intro-DAG}
\end{figure}
\medskip

%We also contrast our results with the generalized Bayes rule
%\begin{align*}\label{GR} \tag{GR}
%(\text{belief at } n+1) \propto  (\text{belief at } n) \times   (\text{Bayesian update at } n+1)^{\gamma},
%\end{align*}
%a special case of  Grether rule \citep{grether} and show that its behavior is fundamentally different. While also underreacting to new information, the updating rule (\ref{GR}) has the same asymptotic properties than Bayesian updating. Underreaction alone does not explain our results. 

\medskip

Our analysis also suggests an advantage to deliberately adopt a misspecified model. In Bayesian statistics, computing posterior distributions is often challenging.  Thus, even when the statistical model is correctly specified, learning the true data-generating process may be computationally intractable -- see the large literature on approximate Bayesian computation. Our analysis suggests an alternative route. Instead of working directly with the distributions of interest, the statistician can restrict attention to families of distributions whose mixtures are rich enough to approximate the distributions of interest, while also yielding posteriors that are easier to compute. Although such a restricted model may be misspecified, applying the updating rule (\ref{C-Bay}) still guarantees convergence to the true data-generating process, without requiring intractable posterior computations.\footnote{As we show later, our main result -- Theorem  \ref{th:limitset} -- also holds when the Bayesian posteriors in rule (\ref{C-Bay}) are replaced with approximations.} For example, restricting attention to Gaussian distributions suffices, since mixtures of Gaussians approximate any distribution arbitrarily well (with respect to the Wasserstein metric).

\medskip 

We close the introduction with a brief discussion of the related literature (Section \ref{sec:lit} contains a detailed discussion).  The two most closely related papers are \cite{Ber66} and \cite{epstein2010non}. \cite{Ber66} studies the long-run properties of Bayesian updating when the model is misspecified and shows that the beliefs converge to the process, which minimizes the KL divergence from the true process.\footnote{ See also  \cite{Fre63}, \cite{Ber70}, \cite{DiaFre86}, \cite{DiaFre99}, \cite{shalizi2009}.} \cite{epstein2010non} study the long-run properties of the updating rule (\ref{C-Bay}) when the model is correctly specified. They show that if the agent under-reacts to new information, the agent eventually learns the true process.  We differ from \cite{Ber66} by assuming non-Bayesian updating and from \cite{epstein2010non} by assuming a misspecified model. Some recent papers, e.g., \cite{bohren2021learning, bohren2024}, \cite{heidhues2021} and \cite{frick2023belief}, provide a link between Bayesian updating in misspecified problems and non-Bayesian updating.  More specifically, it is possible to construct another set of candidate processes such that Bayesian updating in that modified model induces the same beliefs than the  updating rule (\ref{C-Bay}) in the original model. To the best of our understanding,  none of these recent results are applicable to our problem -- see Section \ref{sec:main} for a detailed discussion.

\section{The Problem}\label{sec:model}
At each period $n \in \mathbb{N}$, an agent is observing $x_n \in X$, independently and identically drawn from an \emph{unknown} data-generating process $p^* \in \Delta(X)$. The agent's model is a set $\mathcal{P}:=\{p_{\theta} \in \Delta(X): \theta \in \Theta\}$ of data-generating processes, which may not include the true data-generating process $p^*$ -- in that case, the model is misspecified. Throughout, we assume that the sets $X$ and $\Theta$ are finite and that the support of $p^*$ and each $p_{\theta}$ is $X$.\footnote{None of our results would change if we were to assume that the support of $p^*$ is included in the support of $p_{\theta}$ for all $\theta \in \Theta$. The only modification would be to substitute $X$ with the support of $p^*$ in the analysis.} The latter assumption guarantees that the agent is never surprised, that is, the agent cannot observe a realization $x$, believed  impossible under all data-generating processes in $\mathcal{P}$.\medskip

We can think of $x$ as financial returns in investment problems, losses in insurance problems, or demands in pricing problems. For instance, in a canonical portfolio problem, $x$ corresponds to realized asset returns and $\mathcal{P}$ their postulated distributions, e.g., multivariate Gaussian distributions with unknown mean and variance-covariance matrix (the parameter $\theta$). As another instance, in canonical insurance problems, $x$ corresponds to losses or, equivalently, final wealths. More generally, we view $x$ as payoff-relevant consequences, $p^*$ as their true distribution, and $\mathcal{P}$ as the collection of  ``parameterized'' distributions the agent postulates. An alternative is to view $x$ as a (payoff-irrelevant) signal about the (payoff-relevant) state $\theta$, $p^*$ as the signal distribution when the true state is $\theta^* \in \Theta$, and $(p_{\theta})_{\theta}$ as the perceived conditional distributions (an experiment). Both views are common in Economics and Statistics. While our formal results are independent of either perspective, their interpretations differ. Under the first view, successful learning entails the convergence of the \emph{predictive process} to the true distribution $p^*$, whereas under the alternative view, learning success is characterized by the identification of the true state $\theta^*$. We return to these  alternative interpretations later on.\medskip

The agent's prior is $q_{0} \in \mathbf{S} := \Delta(\Theta)$. Without loss of generality, we assume that the support of $q_0$ is $\Theta$, i.e., $q_0 \in \mathbf{S}^*:=\mathrm{int} \left(\mathbf{S}\right)$, the interior of $\mathbf{S}$. The agent updates his belief from $q_n$ to $q_{n+1}$ upon observing the signal $x_{n+1}$ according to the rule:  
\begin{align*}\label{C-Bay}
q_{n+1} = (1-\gamma_{n+1})q_n + \gamma_{n+1} B (q_n,x_{n+1}), \tag{C-Bay}
\end{align*}
where $B(q_n,x_{n+1}) \in \Delta(\Theta)$ is the Bayesian posterior of $q_n$ given the new observation $x_{n+1}$  
and $\gamma_{n+1}$ the (possibly random and positive) weight on the Bayesian posterior.\footnote{The posterior probability of $\theta$ is $B_{\theta}(q,x) = \frac{q_{\theta} p_{\theta}(x)}{\sum_{\theta'} q_{\theta'} p_{\theta'}(x)}$, whenever the denominator is positive.} The updating rule  (\ref{C-Bay}) was first discussed in \cite{phillips1966} and \cite{edwards1968}, generalizes the Bayesian rule ($\gamma_{n+1}=1$) and can accommodate a number of well-documented  biases such as underreaction ($\gamma_{n+1} < 1$) and overreaction ($\gamma_{n+1} >1  $) to new information, confirmatory biases, and a few others. We refer the reader to \cite{epstein2006} and  \cite{kovach2021} for more extensive discussions and axiomatic foundations. \cite{benjamin2019} documents two additional empirical findings, which the rule can also accommodate. First, the more observations an agent has, the more the agent under-reacts. This finding -- Stylized Fact 2 in \cite{benjamin2019} -- is consistent with decreasing updating weights $(\gamma_n)_n$. Second, individuals' beliefs after sequentially observing $x_1$ to $x_n$ differ from the beliefs they would have formed, had they received the same information simultaneously as a single observation $(x_1,\dots,x_n)$. This finding -- Stylized Fact 8 in \cite{benjamin2019} --  is also consistent with the updating rule (\ref{C-Bay}).\footnote{\cite{cripps2018} characterizes the updating rules, which produce the same update after observing either $x_1$ to $x_n$ sequentially, or $(x_1,\dots,x_n)$ simultaneously.} Throughout, we refer to the process $(\mathbb{E}_{q_n}[p_{\boldsymbol{\theta}}])_{n \in \mathbb{N}}$ under the updating rule (\ref{C-Bay}) as the agent's \emph{predictive} process.

We conclude this section with a brief review of \cite{Ber66} and \cite{epstein2010non}.   \cite{Ber66} studies the convergence of the \emph{Bayesian} posteriors (i.e., $\gamma_n=1, \;\forall n$) and proves the almost-sure convergence to the maximizers of the (negative of the) cross entropy, that is, $  \arg\max_{\theta \in \Theta}\sum_{x}p^*(x)\log p_{\theta}(x)$. \cite{epstein2010non} study whether the updating rule (\ref{C-Bay}) induces the agent to eventually learn the true data-generating process, when the model is \emph{correctly} specified. They show that if the agent is always underreacting to new information (i.e., $\gamma_{n} \leq 1, \forall n$), then the agent eventually learns the true data-generating process, when the weight $\gamma_{n}$ is predictable (measurable with respect to the information up to period $n-1$). With the help of an example, they show that this result is not true when the weight $\gamma_{n}$ is only adapted (measurable with respect to the information up to period $n$).

%\textsc{Example.} The agent observes the number of successes $x$ out of $\overline{x}$ draws and 
%$p_{\theta}$ is a binomial distribution, where $\theta$ is the probability of success, i.e., $p_{\theta}(x)=\binom{\overline{x}}{x}\theta^x(1-\theta)^{\overline{x}-x}$. The key observation to make is that mixture of binomial distributions are not binomial distributions. 

\section{Underreaction and Learning}\label{sec:main}
 
This section presents our main result.  Motivated by Stylized Fact 2 in \cite{benjamin2019},  we assume that the more observations the agent has, the more the agent underreacts to new information.  Yet, we do not want the underreaction to be so severe that the agent stops reacting after finitely many periods. \medskip

Formally, we assume that $(\gamma_n)_n$ is a deterministic sequence of positive real numbers such that 
    \begin{equation} \label{eq:decreasing_step}
     (i): \sum_{n=0}^{+\infty} \gamma_n = + \infty, \text{  and  } \, (ii): \sum_{n=0}^{+\infty}e^{-c/\gamma_n} < +\infty, \; \, \text{  for all} \; c>0.
     \end{equation} 
Examples of series satisfying conditions (i) and (ii) include $(1/n)_{n \in \mathbb{N}}$, $(1/\sqrt{n})_{n \in \mathbb{N}}$, and $(1/\log^2(n))_{n\in \mathbb{N}}$. More generally, the series $(1/n^{\alpha})_n$  with $\alpha \in (0,1]$ and $(1/\log^{\alpha}(n))_n$ with $\alpha >1$ all satisfy the condition.\footnote{A more familiar, but stronger, assumption is $\sum_{n=0}^{+\infty} \gamma_n = + \infty$ and $\sum_{n=0}^{+\infty}\gamma_n^2 < +\infty$.}  Condition (ii) guarantees that the series $(\gamma_n)_n$  converges to zero, so that the more observations the agent has, the less the agent reacts.  Condition (ii) further says $(\gamma_n)_n$  converges to zero faster than the series $(1/\log(n))_n$ does.\footnote{This is not a demanding condition as the series $(1/\log(n))_n$ goes relatively slowly to zero.} Condition (i) guarantees that the agent does not stop reacting after finitely many periods. \medskip

The analysis rests on the observation that the updating rule (\ref{C-Bay}) can be rewritten as:
\begin{align*}
q_{n+1}-q_n =   \gamma_{n+1}\left(H(q_n)+ U_{n+1}\right),
\end{align*}
where  $H: \mathbf{S} \rightarrow \mathbb{R}_{0}^{|\Theta|}$ is defined by
\begin{equation}\label{eq:H}
H_{\theta}(q):=q_{\theta}\Biggl(\underbrace{\sum_{x}p^*(x)\frac{p_{\theta}(x)}{\sum_{\theta'}q_{\theta'}p_{\theta'}(x)}}_{:=f_{\theta}(q)} - 1\Biggr) = \underbrace{\mathbb{E}_{p^*}\left[B_{\theta}(q,\bold{x})\right]}_{\substack{\text{Expectation of Bayesian} \\ \text{posterior of $\theta$}}} - q_{\theta}, \; \forall \theta,
\end{equation}

%\text {Expectation of Bayesian posterior of $\theta$}}- q_{\theta}, \; \forall \theta,
%\end{equation}
and $(U_n)_n$ a bounded martingale difference: $U_{n+1} :=B(q_n,x_{n+1}) - \mathbb{E}_{p^*}[B(q_n,\bold{x}_{n+1})]$.\medskip 

Before proceeding, we remark that the analysis below  continues to apply if the agent updates his belief from $q_{n}$ to  $q_{n+1} :=(1-\gamma_{n+1})q_n + \gamma_{n+1}  b(q_n,x_{n+1})$, where $b(q_n,x_{n+1})$ is a \emph{noisy estimate} of $B(q_n,x_{n+1})$ such that $\mathbb{E}[b(q,x)|B(q,x)] = B(q,x)$ for all $(q,x)$. In words, the updating rule tolerates some errors.      \medskip 

This recursive formulation is a Robbins-Monro algorithm \citep{robbins-monro-51}, a well-known algorithm  in the literature on stochastic approximation theory. A popular method for analyzing such algorithms is the ordinary differential equation (ODE) method, which approximates the sequence $(q_n)_n$ with the solution to the associated (deterministic) ODE:
 \[\dot{q}(t) =H(q(t)).\]
When applicable (which it is in our case), this method reduces the analysis of the Robbins-Monro algorithm (\ref{C-Bay}) to the study of the associated ODE. We prove two central results about the ODE. (All proofs are in the Appendix.) The first result -- Lemma \ref{lem:V-Lyapunov} --  states that the function $q \mapsto V(q)$ is a Lyapunov function for the ODE, where $V(q)$ is the (negative of the) cross-entropy of the mixture $\sum_{\theta}q_{\theta}p_{\theta}$ with respect to $p^*$, that is, 
\begin{align}\label{eq:V}
V(q):= \sum_x p^*(x) \log\left(\sum_{\theta}q_{\theta}p_{\theta}(x) \right). 
\end{align} 
\medskip 

The second result -- Lemma \ref{lem:decomposition} -- states that the set $E:= \left\{q \in \mathbf{S}: \; \, H(q) = 0 \right\}$ of zeroes of $H$ can be decomposed into a finite union of  disjoint, compact and \emph{convex} components $C_{k}$, $k=1,\dots,K$, with the cross-entropy $V$ constant on each of the components.\footnote{Note that the set of degenerated beliefs $\{\delta_{\{\theta\}}: \theta \in \Theta\}$ belong to $E$, so that $E$ is non-empty.} The decomposition is unique, and $V$ attains its maximal value on a unique component, which we denote $C_{k^*}$: 
 \[V(C_{k^*}) > V(q), \text{ for all } q \in \mathbf{S} \setminus C_{k^*}.\]

To get some intuition for Lemma \ref{lem:decomposition}, let $C_{\widehat{\Theta}}$ be the maximizers of the cross-entropy $V$ over the face $F_{\widehat{\Theta}}$, where $F_{\widehat{\Theta}}:=\{q \in \mathbf{S}: q_{\theta} > 0 \text{\;iff\;} \theta \in \widehat{\Theta}\}$. If non-empty, $C_{\widehat{\Theta}}$ is convex and  induces a unique $p \in \Delta(X)$, so that the cross-entropy is constant on $C_{\widehat{\Theta}}$.   Since there are finitely many faces, there are finitely many non-empty sets of maximizers. (At least one is non-empty, since $V$ is continuous and $\mathbf{S}$ non-empty and compact.) Moreover, it is easy to see that each maximizer is a zero of $H$ and, conversely. Intuitively, $H(q)=0$ are the Kuhn-Tucker conditions for the maximization of the concave function $V$ with respect to $q \in F_{\widehat{\Theta}}$.  We then construct each $C_k$  as a carefully chosen union of maximizers. Example \ref{ex:convex-components-maximal-number} illustrates the construction. Note that we can have as many components as we have faces.

\begin{example}\label{ex:convex-components-maximal-number} Assume that  $\Theta=\{\theta_1,\theta_2,\theta_3\}$, $ X=\{x_1,x_2,x_3\}$, and 
    \[p_{\theta_1} = \frac{1}{5}(1,2,2), \; p_{\theta_2} = \frac{1}{5}(2,1,2), \; p_{\theta_3}= \frac{1}{5}(2,2,1), \, p^*= \frac{1}{3}(1,1,1).\]
The set $E$ is:
\[\left\{\delta_{\{\theta_1\}}\right\} \cup \left\{\delta_{\{\theta_2\}}\right\} \cup \left\{\delta_{\{\theta_3\}}\right\} \cup \left\{\frac{1}{2}(1,1,0)\right\} \cup  \left\{\frac{1}{2}(1,0,1)\right\} \cup  \left\{\frac{1}{2}(0,1,1)\right\}\cup \left\{\frac{1}{3}(1,1,1) \right\}.\]
We have $V(\delta_{\{\theta_1\}}) = V(\delta_{\{\theta_2\}}) = V(\delta_{\{\theta_3\}}) = \frac{1}{3} \log \frac{1}{5} + \frac{2}{3} \log \frac{2}{5}$,
\[V\left(\frac{1}{2}(1,1,0)\right) = V\left(\frac{1}{2}(1,0,1)\right)= V\left(\frac{1}{2}(0,1,1)\right)= \frac{2}{3} \log \frac{3}{10} + \frac{1}{3} \log \frac{4}{10},\] and $V\left(\frac{1}{3}(1,1,1)\right) = \log \frac{1}{3}$. Here, $C_{k^*} = \left\{\frac{1}{3}(1,1,1) \right\}$, the interior point.

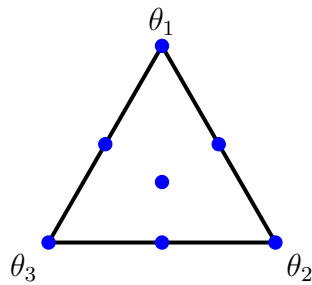
\begin{figure}[h!]
\centering
 \usetikzlibrary{arrows,shapes.geometric,decorations.markings, calc, fadings, decorations.pathreplacing, patterns, decorations.pathmorphing, positioning}
\newcommand{\midarrow}{\tikz \draw[-triangle 90] (0,0) -- +(.2a5,0);}
\begin{tikzpicture}
    \draw [black,line width=1.5pt] (0,0) -- (1.5,2.6) -- (3,0) -- cycle;
    %\draw [blue,line width=1pt] (0.75,1.3) -- (1.8,0);
    \filldraw[blue] (0,0) circle (2.5 pt);
    \filldraw[blue] (3,0) circle (2.5 pt);
    \filldraw[blue] (1.5,2.6) circle (2.5 pt);
    \filldraw[blue] (0.75, 1.3) circle (2.5 pt);
    \filldraw[blue] (1.5, 0) circle (2.5 pt);
    \filldraw[blue] (2.25, 1.3) circle (2.5 pt);
    \filldraw[blue] (1.5, 0.8) circle (2.5 pt);
    \node [below left] at (0,0) {$\theta_3$};
    \node [above] at (1.5,2.6) {$\theta_1$};
    \node [below right] at (3,0) {$\theta_2$};
\end{tikzpicture}
\caption{The components in Example \ref{ex:convex-components-maximal-number}}
\end{figure}
\end{example}
\medskip

To sum up, the cross-entropy $V$ is a Lyapunov function for the flow of the dynamics and maximized in a unique component $C_{k^*}$ of zeroes. 
This suggests that the component $C_{k^*}$ is globally stable.  While our main result -- Theorem \ref{th:limitset} -- states that this is indeed the case, the result is far from immediate. As we shall see, the main technical challenge is that the components $C_k$ are not always \emph{singletons} (isolated points) and, moreover, often exist as \emph{continuums} on the boundary. 

\begin{theorem}\label{th:limitset}
For all $q_0 \in \mathrm{int} \left(\mathbf{S}\right)$,  the (random) limit set $\mathcal{L}((q_n)_n)$ of the process $(q_n)_n$ is contained in  $C_{k^*}$, with probability one. In words, the beliefs converge to maximizers of the cross-entropy $V$.   
\end{theorem}

Theorem \ref{th:limitset}  states that every $\omega$-limit point $q$ induces a unique distribution $\sum_{\theta}q_{\theta}p_{\theta} \in \co \mathcal{P}$, which attains the maximum:
\begin{align*}
\max_{p \in \co\mathcal{P}} \sum_{x}p^*(x)\log p(x).
\end{align*}
It then follows immediately that: 

\begin{corollary}
If $p^* \in  \co\mathcal{P}$, then the agent's predictive process converges to the true data-generating process $p^*$.
\end{corollary}

Thus, if $p^* \in  \co\mathcal{P}$, the learning outcome is the same as the one of a Bayesian learner with the richer set of (non-parametric) models $\co \mathcal{P}$, assuming that the agent's prior assigns a positive probability to an open neighborhood of the cross-entropy maximizer -- see \cite{Ber66}.  This observation is reminiscent of  \cite{bohren2024}, who demonstrate an equivalence between non-Bayesian updating in a correctly specified model and Bayesian updating in a misspecified model. In particular, if the agent's model is $\{\pi_{\theta}^{q_n, \gamma_n}: \theta \in \Theta\}$, with
\begin{align*}
\pi_{\theta}^{q_n,\gamma_n}(x):= (1-\gamma_n)\mathbb{E}_{q_n}[ p_{\boldsymbol{\theta}}(x)] + \gamma_np_{\theta}(x)
\end{align*}
for all $x$, then the Bayesian posteriors of the modified model $\{\pi_{\theta}^{q_n, \gamma_n}: \theta \in \Theta\}$ coincide with the non-Bayesian posteriors (\ref{C-Bay}) of the original model $\{p_{\theta}: \theta \in \Theta\}$, when the agent's prior belief is $q_n$ at period $n$. We stress that the ``\textit{as-if}'' model depends on the current belief $q_n$ and weight $\gamma_n$.  The dependence on current beliefs is a common feature of social learning models, see e.g., \cite{esponda2016berk}, \cite{bohren2021learning}, and \cite{frick2023belief}. However, as far as we know, none of the results developed in these (and other related) papers are directly applicable to our problem. Indeed, the main focus of this literature is on the stability of degenerated beliefs and none of the few instances, where mixed beliefs are considered, are directly applicable in our setting. See Section \ref{sec:lit} for a more detailed discussion.

\medskip 

Conservative updating thus protects the agent against misspecifications and leads to \emph{better predictions}.  Predicting better, even perfectly,  does not imply that the agent learns (identifies) the true data-generating process. In fact, the agent \emph{cannot} learn the true data-generating process when his model is misspecified! Yet, when we interpret $\Theta$ as (payoff-relevant) states and $X$ as signals, an implication of Theorem \ref{th:limitset} is that the agent's belief about the states may also come closer to the true state than Bayesian updating would. For an illustration, suppose that there are two states $\theta^*$ and $\theta$ and two signals $\ell$ and $h$. The true conditional distributions (experiment) are $p^*_{\theta^*}(h)=1-p^*_{\theta}(h)=2/3$, while the agent's model is $p_{\theta^*}(h)=1-p_{\theta}(h)=1/10$. Assume that the true state is $\theta^*$. Since the agent wrongly interprets the signal $h$ as evidence of $\theta$ and the signal $h$ is the most likely when the state is $\theta^*$, the Bayesian posterior converges almost surely to $\theta$, the wrong state. In contrast, with the updating rule (\ref{C-Bay}), the agent's belief converges to probability $7/24$ on $\theta^*$, a closer distribution to the truth. (Its cross-entropy with respect to $\delta_{\{\theta^*\}}$ is $\log (7/24)$, while the cross-entropy of the limiting Bayesian posterior is $- \infty$.\footnote{We compare distributions in $\mathbf{S}$. Since the parameter identifies distributions in $\Delta(X)$, we are equivalently comparing distributions in $\Delta(\Delta(X))$.}) It is not always true that the non-Bayesian agent learns better, though. E.g., if we change $1/10$ for $9/10$, the agent learns the true state under Bayesian updating, while the belief converges to $17/24$ on $\theta^*$ under updating (\ref{C-Bay}). In general, all we can say is that the agent's belief comes closer, and sometimes strictly so, to the true state under updating (\ref{C-Bay}) (in the sense of the Kullback-Leibler divergence) if, and only if, the agent does not learn the true state under Bayesian updating. 

\medskip 

The updating rule (\ref{C-Bay}) may also serve as a ``\emph{misspecification test}.'' Indeed, if the process $(q_n)_n$ does not converge to a point mass on one of the parameter $\theta$, then the model $\mathcal{P}$ must be misspecified.\footnote{More precisely, we need to add the mild requirement that the model is identified, i.e., $p_{\theta} \neq p_{\theta'}$ for all $(\theta,\theta')$. When $p^*=p_{\theta}=p_{\theta'}$ with $\theta \neq \theta'$, it is possible for the belief process to converge to an interior point, even though the model is correctly specified.} Thus, a possible test is to reject the hypothesis that the model is correctly specified if  $\min_{\theta \in \Theta}||q_n - \delta_{\{\theta\}}||_{TV} > \varepsilon$, where $\varepsilon$ and $n$ are carefully calibrated to trade-off type I and II errors. We leave this question for future research.

\medskip

The proof of Theorem \ref{th:limitset} uses results from stochastic approximation theory with, however, one major difficulty.   The major difficulty is that the convex components of $E$ might  neither be isolated points nor interior points. In fact, each component $C_k \neq C_{k^*}$ lies on the boundary of $\mathbf{S}$.  Example \ref{ex:continuum-attractors} illustrates that we may have a continuum of zeroes on the boundary.

\begin{example}\label{ex:continuum-attractors} Suppose that $\Theta= \{\theta_1,\theta_2,\theta_3,\theta_4\}, X=\{x_1,x_2,x_3\}$, and  
    \[p_{\theta_1} = \frac{1}{4}(2,1,1), \; p_{\theta_2} = \frac{1}{4}(1,2,1), \; p_{\theta_3}= \frac{1}{4}(1,1,2), \;  p_{\theta_4}= \frac{1}{8}(2,3,3),p^* =\frac{1}{8}(6,1,1). \]
    The set $E$  is 
\[\{\delta_{\{\theta_1\}}\} \cup  \{\delta_{\{\theta_2\}}\} \cup \{\delta_{\{\theta_3\}}\} \cup \{(0,\lambda/2,\lambda/2,1-\lambda): \; \lambda \in [0,1]\}, \]
with $C_{k^*} = \{\delta_{\{\theta_1\}}\}$. The component $\{(0,\lambda/2,\lambda/2,1-\lambda): \; \lambda \in [0,1]\}$ is a continuum and lies on the boundary of $\mathbf{S}$ -- it is the line in Figure \ref{fig:continuum-attractors}.
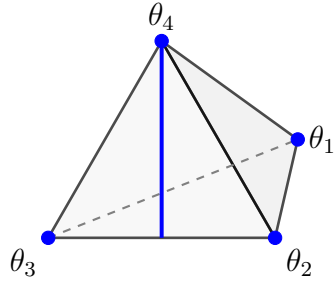
\begin{figure}[h!]
\centering
 \usetikzlibrary{arrows,shapes.geometric,decorations.markings, calc, fadings, decorations.pathreplacing, patterns, decorations.pathmorphing, positioning}
\newcommand{\midarrow}{\tikz \draw[-triangle 90] (0,0) -- +(.2a5,0);}
\begin{tikzpicture}
    \fill[gray!10, opacity=0.5] (0,0) -- (1.5,2.6) -- (3,0) -- cycle;
        \fill[gray!20, opacity=0.5]  (1.5,2.6) -- (3.3, 1.3) -- (3,0) -- cycle;
         \draw[black, line width=1pt, opacity=0.7] (0,0) -- (1.5,2.6) -- (3,0) -- cycle;
    \draw[black, line width=1pt, opacity=0.7] (1.5,2.6) -- (3.3,1.3) -- (3,0) -- cycle;
\draw[dashed, gray, line width=0.8pt] (0,0) -- (3.3,1.3);
        
    \draw [blue,line width=1.5pt] (1.5,2.6) -- (1.5,0);
    \filldraw[blue] (0,0) circle (2.5 pt);
    \filldraw[blue] (3,0) circle (2.5 pt);
    \filldraw[blue] (1.5,2.6) circle (2.5 pt);
        \filldraw[blue] (3.3, 1.3) circle (2.5 pt);
        \node[right] at (3.3, 1.3) {$\theta_1$};
    \node [below left] at (0,0) {$\theta_3$};
    \node [above] at (1.5,2.6) {$\theta_4$};
    \node [below right] at (3,0) {$\theta_2$};
\end{tikzpicture}
\caption{The components in Example \ref{ex:continuum-attractors}}
\label{fig:continuum-attractors}
\end{figure}
\end{example}

As a consequence, none of the classical results in stochastic approximation theory, e.g., \cite{pemantle1990},  \cite{brandiere1998}, \cite{brandiere-duflo-1996}, are immediately applicable. For instance, \citet[p. 700-701]{pemantle1990} states that ``\textit{if there are uncountably many unstable points, then $(q_n)$ can still converge to one of them even though each point has probability zero of being hit. In this case, all one might conclude is that the law of the limit has no point masses at unstable critical points\footnote{Here, the term ``critical points" is relative to the flow. In other terms, a critical point of the flow associated to the EDO $\dot{q} = H(q)$ is simply a zero of $H$.}.}'' Similarly, since the components $C_k$, $k \neq k^*$, are on the boundary of $\mathbf{S}$, we cannot rely on the  Jacobian of $H$ in the neighborhood of each $q \in C_k$ to determine their stability. \medskip 

We develop novel arguments to tackle that issue. The arguments are quite technical, but the main logic is simple: We show that each component $C_k \neq C_{k^*}$ admits an open neighborhood and a \emph{uniform repelling} direction, so that the dynamics cannot hop indefinitely often from elements of the open neighborhood of $C_k$ to others.\medskip

 We now present the arguments for the special case where $C_k \neq C_{k^*}$ consists of the single element ${\hat{q}}$.  (Technical details in this paragraph may be omitted without loss of continuity.)  Choose $q^* \in C_{k^*}$ with maximal support, that is, for all $q \in C_{k^*}$, $[q_{\theta} > 0 \Rightarrow q^*_{\theta}>0]$. The \emph{convexity} of each component guarantees that such $q^*$ exists.\footnote{This is where having convex, and not merely connected, components helps.} We partition $\Theta$ into three subsets: $\Theta_1:= \{\theta \in \Theta: \; q^*_{\theta} = \hat{q}_{\theta} = 0\}$, $\Theta_2:= \{\theta \in \Theta: \; \hat{q}_{\theta} >0 \}$, and $\Theta_3:= \left\{\theta \in \Theta: \; \hat{q}_{\theta} =0, q^*_{\theta}>0  \right\}$. From the construction of the components $C_k$,  $\Theta_3$ is non-empty. (If  $\Theta_3$ was empty, then $C_k=C_{k^*}$.) From the strict   concavity  of $V$,  $\frac{\partial}{\partial \tau}V\left(\tau q^* + (1-\tau)\hat{q} \right) >0$, $\tau \in (0,1)$. Therefore,
\[0< \sum_{\theta \in \Theta} (q^*_{\theta} - \hat{q}_{\theta}) f_{\theta}(\hat{q}) =  \sum_{\theta \in \Theta} q^*_{\theta} f_{\theta}(\hat{q}) - 1 =  \sum_{\theta \in \Theta_2 \cup \Theta_3} q^*_{\theta} f_{\theta}(\hat{q}) - 1  = \sum_{\theta \in \Theta_2} q^*_{\theta} +  \sum_{\theta \in \Theta_3} q^*_{\theta}  f_{\theta}(\hat{q}) - 1,\]
where the first and third equalities follows from $\hat{q}_{\theta}(f_{\theta}(\hat{q})-1)=0$ for all $\theta$, since $\hat{q}$ is a zero of $H$. (See Equation \ref{eq:H} for the definition of $f_{\theta}(\hat{q})$.) Since $\sum_{\theta \in \Theta_2 \cup \Theta_3} q_{\theta}^* = 1$, there exists $\hat{\theta} \in \Theta_3$ such that $f_{\hat{\theta}}(\hat{q}) >1$, that is, $\dot{q}_{\hat{\theta}} > 0$ in the neighborhood of $\hat{q}$ -- the dynamics cannot converge to $\hat{q}$. The logic extends to arbitrary convex components. 

\medskip

The observation that the components are not necessarily isolated points is a direct consequence of the fact that distributions in $\co \mathcal{P}$ do not always have a unique decomposition in terms of its extreme points $\mathcal{P}$. For future reference, let $\Lambda$ be the set of all possible convex combinations of the $p_{\theta}$'s equal to $p^*$, that is, 
\[\Lambda := \left\{\lambda = (\lambda_{\theta})_{\theta \in \Theta} \in \mathbf{S}: \; \, p^* = \sum_{\theta \in \Theta} \lambda_{\theta} p_{\theta} \right\}.\]
We now discuss the implications of assuming different notions of independences of the family $\mathcal{P}$.  We say that the family $\mathcal{P} = \{p_{\theta}: \theta \in \Theta\}$ is \emph{full} if the vectors $(p_{\theta}-p_{\theta'})_{(\theta,\theta') \in \Theta \times \Theta}$ span the tangent space associated to $\Delta(X)$, i.e., $\mathbb{R}^{|X|}_0:=\{v \in \mathbb{R}^{|X|}: \, \; \sum_{x \in X } v_x  = 0\}$. In words, if the family is full, then differences between its members can point in any direction allowed by the simplex structure. This property is common in statistics, as it allows for identification.

\begin{proposition} \label{prop:full}
  If the family $\mathcal{P}$ is full, then $E \cap \mathbf{S}^* \subseteq \Lambda$. If, in addition, $p^* \notin \co \mathcal{P}$, then $E \subseteq  \partial \left(\mathbf{S}\right)$. 
 \end{proposition}
 
 Thus, when the family is full, interior zeroes, if they exist, ``identify'' the true data-generating process $p^*$; the agent would learn to predict correctly.  We stress that the inclusion $E \cap \mathbf{S}^* \subseteq \Lambda$ may be strict, because $\Lambda$ may intersect the boundary of $\mathbf{S}$, as the following example demonstrates. 

 \begin{example}\label{ex:full-family} Suppose that $\Theta= \{\theta_1,\theta_2,\theta_3\}, X=\{x_1,x_2\}$, and  
    \[p_{\theta_1} = \frac{1}{4}(1,3); \; p_{\theta_2} = \frac{1}{3}(1,2); \; p_{\theta_3}=\frac{1}{4}(3,1),  p^*=\frac{1}{2}(1,1).\]
   Then, $E = \{\delta_{\{\theta_1\}}\} \cup \{\delta_{\{\theta_2\}}\} \cup \{\delta_{\{\theta_3\}}\} \cup \Lambda$, where {\small
    \[\Lambda = \left\{\lambda \in \mathbf{S}: \, \lambda_1 p_{\theta_1} + \lambda_2 p_{\theta_2} + \lambda_3 p_{\theta_3} = p^* \right\} = \left\{\left(\lambda_1, \frac{1}{5}(3-6\lambda_1),\frac{1}{5}(2 + \lambda_1) \right): \; \lambda_1 \in \left[0,\frac{1}{2}\right] \right\}.\]}

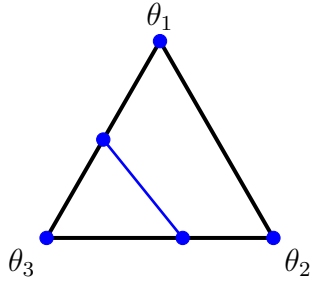
\begin{figure}[h!]
\centering
 \usetikzlibrary{arrows,shapes.geometric,decorations.markings, calc, fadings, decorations.pathreplacing, patterns, decorations.pathmorphing, positioning}
\newcommand{\midarrow}{\tikz \draw[-triangle 90] (0,0) -- +(.2a5,0);}
\begin{tikzpicture}
    \draw [black,line width=1.5pt] (0,0) -- (1.5,2.6) -- (3,0) -- cycle;
    \draw [blue,line width=1pt] (0.75,1.3) -- (1.8,0);
    \filldraw[blue] (0,0) circle (2.5 pt);
    \filldraw[blue] (3,0) circle (2.5 pt);
    \filldraw[blue] (1.5,2.6) circle (2.5 pt);
     \filldraw[blue] (0.75,1.3) circle (2.5 pt);
      \filldraw[blue] (1.8,0) circle (2.5 pt);
    \node [below left] at (0,0) {$\theta_3$};
    \node [above] at (1.5,2.6) {$\theta_1$};
    \node [below right] at (3,0) {$\theta_2$};
\end{tikzpicture}
\caption{The four components in Example \ref{ex:full-family}}
\end{figure}
\end{example}

We say that the family $\mathcal{P}$ satisfies convex independence  if there does not exist $\hat{\theta} \in \Theta$ such that $p_{\hat{\theta}} \in \co \{p_{\theta}: \theta \in \Theta \setminus \{\hat{\theta}\}\}$. The next result states that  all components are singletons when the family satisfies convex independence and $|\Theta|\ \leq |X|$. 

\begin{proposition}\label{prop:convex-indep}
    If $|\Theta| \leq |X|$ and the family $\mathcal{P}$ satisfies convex independence, then every $C_k$ is a singleton.
\end{proposition}

Finally, we say that the family $\mathcal{P}$ is \emph{tight} if it is full and \emph{minimal}, in the sense that any proper sub-family is not full.  Note that a family is  tight if and only if $|\Theta| = |X|$ and it satisfies convex independence. We have the following theorem when the family is tight. 

\begin{theorem}
    Assume that the family $\mathcal{P}$ is tight. Then, there exists a unique $\lambda^*$ such that  $(q_n)_n$ converges to $\lambda^*$ almost surely. If $\Lambda \neq \emptyset$, then $\Lambda = \{\lambda^*\}$.
    \end{theorem}

\section{Discussion}
\subsection{Updating Weights}\label{sec:weights}

Section \ref{sec:main} assumes that the updating weights $(\gamma_n)_n$ are deterministic and vanishing not too slowly. In line with some empirical evidences \citep{benjamin2019}, these assumptions guarantees that the agent does not stop learning, but does so at a progressively slower rate.  This section discusses alternative assumptions and their consequences.\medskip

We start with an observation. The updating weights can be random: Theorem \ref{th:limitset} remains true if  each random weight $\boldsymbol{\gamma}_n$ is measurable with respect to $\mathcal{F}_{n-1}$, $\mathbb{E}[\sum_n \boldsymbol{\gamma}_n] = +\infty$, and $\mathbb{E}[\sum_n e^{-c/\gamma_n}] < +\infty$, for all $c>0$.\footnote{ We denote $(\mathcal{F}_n)_n$ the filtration adapted to the sequence $(q_n,x_n)_n$.} See Remark 4.3  of \cite{Ben99}.  We now turn our attention to two more substantive alternatives.\medskip

\textbf{\textsc{Constant Weights.}} Assume that the updating weights are constant, i.e.,  $\gamma_n =\gamma \in (0,1)$ for all $n$. With a constant weight, the agent continues to underreact to information, but does so at a constant rate.  The analysis in \cite{epstein2010non} covers this case and shows that the agent eventually learns the data-generating process, that is, if $p_{\theta^*}$ is the  true data-generating process, the agent's beliefs converge almost surely to $\delta_{\{\theta^*\}}$. The argument is  standard. Under the true data-generating process $p_{\theta^*}$, the random process $(\log q_n(\theta^*))_n$ is a sub-martingale, bounded from above, and thus converges almost surely by the martingale convergence theorem. To see that $(\log q_n(\theta^*))_n$ is a sub-martingale, note that 
\begin{align*}
\mathbb{E}_{p_{\theta^*}}[\log q_{n+1}(\theta^*)] - \log q_{n}(\theta^*) & = \mathbb{E}_{p_{\theta^*}}\left[\log\left((1-\gamma)+ \gamma \frac{p_{\theta^*}(x)}{\sum_{\theta'}q_{n}(\theta')p_{\theta'}(x)}\right)\right]\\ 
& \geq \gamma \sum_{x} p_{\theta^*}(x)\log\left(\frac{p_{\theta^*}(x)}{\sum_{\theta'}q_{n}(\theta')p_{\theta'}(x)}\right) \geq 0,
\end{align*}
where the first inequality follows from the concavity of $\log$ and the second from the positivity of the relative entropy. It follows that    $(q_n(\theta^*))_n$ converges almost surely to one.\medskip 

  However, when the model is misspecified (i.e., $p^* \notin \mathcal{P}$), the argument does not apply. Indeed, when the model is misspecified,  the
term
\begin{align*}
 \gamma \sum_{x} p^*(x)\log\left(\frac{p_{\theta^*}(x)}{\sum_{\theta'}q_{n}(\theta')p_{\theta'}(x)}\right), 
\end{align*}
is not always positive. For some $q$, the predictive distribution $\sum_{\theta}q_{\theta}p_{\theta}$ is strictly closer to $p^*$ than $p_{\theta^*}$ is, in which case the term above is strictly negative. \medskip

In fact, the belief process is not even guaranteed to converge. For an example, suppose that $X = \{0,1\}$, $\Theta=\{\theta,\theta'\}$, $p^*(1)=2/3$, $p_{\theta}(1)=3/4 = p_{\theta'}(0)$. We plot the belief processes in Figure \ref{fig:constant-weight}. The blue curve represents the predictive process induced by the updating rule (\ref{C-Bay}) with a constant weight, while the black curve represents the true data-generating process. Clearly, the process does not converge, but oscillates around the true data-generating process. For completeness, we also plot the predictive process induced by the updating rule (\ref{C-Bay}) with decreasing weights (resp., Bayesian updating) as the green (resp., purple) curve. 

\begin{figure}[ht!]
\centering
\includegraphics[scale=0.8]{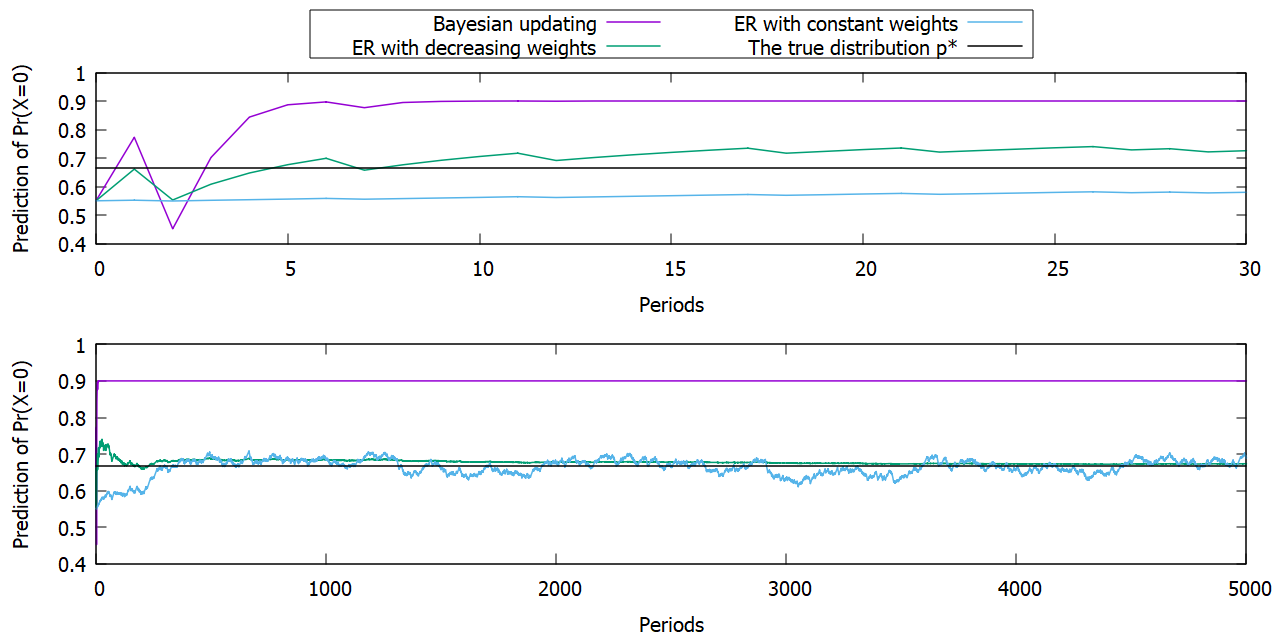}
\caption{Constant weight and non-convergence}\label{fig:constant-weight}
\end{figure}

\medskip

We therefore follow a different approach than in the previous section. We start by noting  that the updating rule (\ref{C-Bay}) induces a Markov chain on $\mathbf{S}$,  where the probability to transition from $q$ to $q'$ is $\sum_{x \in X_{q,q'}} p^*(x)$, $X_{q,q'} := \{x \in X: q' = (1-\gamma) q  + \gamma B(q,x)\}$. From any $q$, the Markov chain transitions to  finitely many $q'$. We then study the long-run properties of the Markov chain, which hold for almost all possible realizations of the observations. More precisely, we study the  convergence properties of the sequences of occupation measures $(\Pi^{\gamma}_n)_n$, where for any Borel set $A \subset \mathbf{S}$,
\begin{align*}
\Pi^{\gamma}_n (A):=\frac{1}{n}\sum_{m=1}^{n} \delta_{q_m}(A).
\end{align*} 
The notation stresses that the chain and, therefore, its occupation measure is parameterized by $\gamma$.  

\medskip

We first prove that the sequences of occupation measures converge to invariant distributions $\pi \in \Delta(\Delta (\Theta))$ of the Markov chains. This result is, however, of limited value as there are multiple invariant measures. For instance, for each $\theta \in \Theta$, $\delta_{\{\delta_{\{\theta\}}\}}$ is an invariant measure. We, however, prove that as  $\gamma$ becomes arbitrarily small,  the occupation measure converges to an invariant measure, which concentrates on $C_{k^*}$. (Recall that $C_{k^*}$ is the component of zeroes, which maximize the cross-entropy.) Thus, when the reaction to new information is small enough, the process is qualitatively similar to the process with decreasing weights studied in Section \ref{sec:main}. We now turn to a formal definition and statement.

\begin{definition}\label{def:limiting measure}
The measure $\pi^* \in \Delta(\mathbf{S})$ is a limiting measure for the updating rule (\ref{C-Bay}) if there exists a sequence $(\gamma_{\ell},\pi_{\ell})_{\ell \in \mathbb{N}}$ such that:
\begin{enumerate}
\item $ \gamma_{\ell} \downarrow 0$,
\item $\pi_{\ell}$ is a weak* limit point of  $(\Pi_n^{\gamma_{\ell}})_n$, for all $\ell$, and
\item $\lim_{\ell} \pi_{\ell} = \pi^*$ for the weak* topology. 
\end{enumerate}
\end{definition}

In words, the measure $\pi^*$ is a limiting measure if we can find a sequence of updating weights $(\gamma_{\ell})_{\ell}$ and a corresponding sequence $(\pi_{\ell})_{\ell}$ of limit points of $(\Pi_n^{\gamma_{\ell}})_n$, which converges to $\pi^*$ for the weak* topology.  We then prove the following:

\begin{theorem}\label{th:constant-step}
Let $\pi^*$ be a limiting measure for the updating rule (\ref{C-Bay}). The support of $\pi^*$ is included in $C_{k^*}$. 
\end{theorem}

Theorem \ref{th:constant-step} states that the limiting measure concentrates on $C_{k^*}$ as the updating weight becomes smaller, the natural analogue of Theorem \ref{th:limitset}.  To get some intuition, recall that a limit point $\pi_{\ell}$ of $(\Pi_n^{\gamma_{\ell}})_n$ is an invariant measure of the Markov chain, parameterized by $\gamma_{\ell}$.  Since invariant measures capture the long-run stationary properties  of the belief process, beliefs should not grow at the limit. More precisely, let
\begin{align*}
r_{\theta}^{\gamma}(q) := \sum_{x} p^*(x) \underbrace{\log\left( (1-\gamma) + \gamma \frac{p_{\theta}(x)}{\sum_{\theta'}q_{\theta'}p_{\theta'}(x)}\right)}_{\approx \frac{q_{n+1}(\theta |x)-q_n(\theta)}{q_{n}(\theta)}}. 
\end{align*}
be the expected growth rate of the belief in $\theta$, when the current belief is $q$. We show that  
\begin{align}\label{eq:growth-rate}
\int_{\mathbf{S}}\sum_{\theta} \lambda_{\theta}r_{\theta}^{\gamma_{\ell}}(q)\pi_{\ell}(dq) \leq 0,
\end{align}
for all positive weight $\lambda \in \mathbb{R}_{+}^{|\Theta|}$, for all limit point $\pi_{\ell}$. Now, choose $\lambda^* \in C_{k^*}$. For any $q \notin C_{k^*}$, we have that 
\begin{align*}
\sum_{x}p^*(x)\left(\frac{\sum_{\theta'}\lambda_{\theta'}^*p_{\theta'}(x)}{\sum_{\theta'}q_{\theta'}p_{\theta'}(x)}-1\right) \geq \sum_{x}p^*(x)\log\left(\frac{\sum_{\theta'}\lambda_{\theta'}^*p_{\theta'}(x)}{\sum_{\theta'}q_{\theta'}p_{\theta'}(x)}\right) >0,
\end{align*}
where the first inequality follows from the concavity of the $\log$ and the second from the definition of $C_{k^*}$. The left-hand side is the derivative of $\gamma \rightarrow \sum_{\theta}\lambda_{\theta}^*r_{\theta}^{\gamma}(q)$ evaluated at $\gamma=0$. Hence, for all $q \notin C_{k^*}$, there exists $\gamma(q) >0$ such that for all $\gamma \leq \gamma(q)$, $\sum_{\theta}\lambda_{\theta}^*r_{\theta}^{\gamma}(q)>0$.  Now, if we had a uniform $\gamma_0>0$ such that  $\sum_{\theta}\lambda_{\theta}^*r_{\theta}^{\gamma}(q)>0$ for all $q \notin C_{k^*}$, for all $\gamma \leq \gamma_0$, and Equation (\ref{eq:growth-rate}) was holding in equality, this would prove that no invariant distribution can put strictly positive probability outside the set $C_{k^*}$ for all $\gamma$ small enough. For instance, when $C_{k^*}$ is the singleton $\{\lambda^*\}$, this would say that $\lambda^*$ is an invariant distribution of the Markov chain for all $\gamma$ small enough. This is not true in general. What we can show, however, is that the same logic applies to appropriately chosen neighborhoods of $C_{k^*}$.   Choosing a sequence of such neighborhoods converging to $C_{k^*}$ then completes the proof. In simple terms, for small enough $\gamma_{\ell}$, the invariant measure $\pi_{\ell}$ puts most of its mass on $C_{k^*}$, but not all of it.  
\medskip

\textbf{\textsc{Observation-dependent Weights.}} Biases such as the self-confirmation bias require the weights of the updating rule (\ref{C-Bay}) to depend on the realized signal $x_{n}$ at period $n$ and, possibly, the belief at period $n-1$ (see, e.g., \cite{rabin1999}). Now, if each $\gamma_{n}(x,q)$ is an arbitrary function of $(x,q)$, there is little hope in characterizing the belief process. As  \cite{epstein2010non} have already shown, there are instances, where the process converges to a Dirac on a wrong state, even when the model is correctly specified. Similarly, as the blue curve in Figure \ref{fig:obs-dep-weight} shows, the process may oscillate perpetually.\footnote{To draw Figure \ref{fig:obs-dep-weight}, we have assumed that $X = \{0,1\}$, $\Theta=\{\theta,\theta'\}$, $p^*(1)=2/3$, $p_{\theta}(1)=9/10 = p_{\theta'}(0)$, and $\gamma_n(x,q)=1$ if  $(B(x,q)-q)(1/2-q)>0$, and $1/n$, otherwise, with $q \in [0,1]$ the belief about $\theta$. In words, the agent updates as a Bayesian when the observation contradicts his belief and is conservative, otherwise.} 

\begin{figure}[h!]
\centering
\includegraphics[scale=0.8]{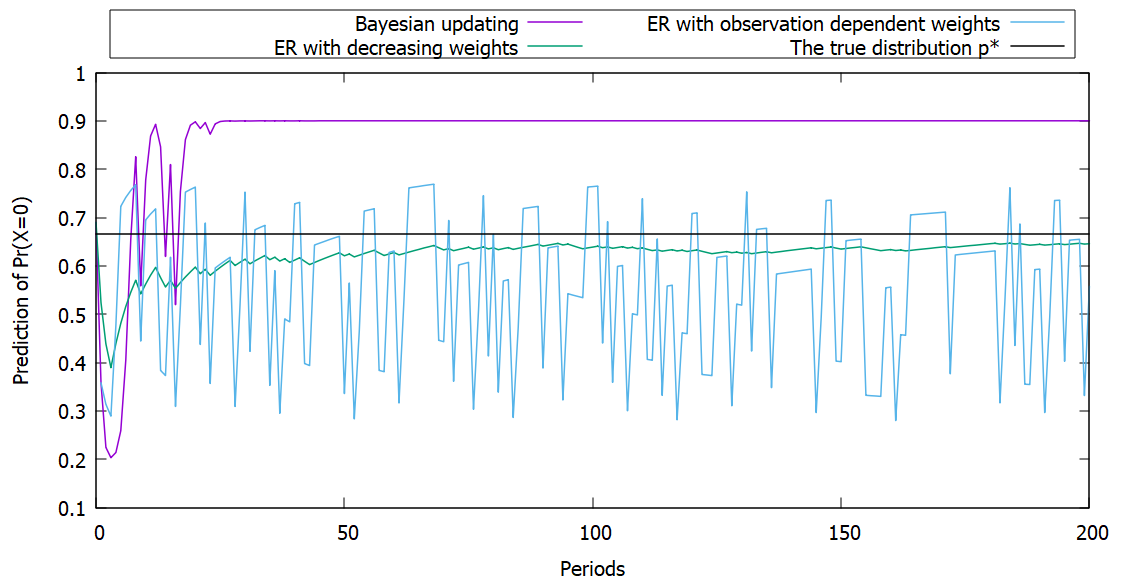}
\caption{Observation-dependent weights and oscillation}
\label{fig:obs-dep-weight}
\end{figure}

There is, however, a special case we can deal with.  Assume that the weights can be written as 
$\gamma_n(x)=\gamma_n \times \alpha(x)$, where $\alpha(x) \in (0,1]$ for all $x \in X$. With this formulation, the weights are independent of the current belief. In addition, the relative weights $\gamma_n(x)/\gamma_n(y)$ are independent of the number $n$ of observations. If the agent reacts more to the observation $x$ than $y$, he continues to do so throughout. With this formulation, we can rewrite the Robbins-Monro algorithm as:
\begin{align*}
q_{n+1}-q_n = \gamma_{n+1}\left(\widehat{H}(q_n) + \widehat{U}_{n+1}\right),
\end{align*}
where $\widehat{H}$ is the function from  $\mathbf{S}$ to $\mathbb{R}^{|\Theta|}$, whose $\theta$-coordinate is 
\begin{align*}
\widehat{H}_{\theta}(q):=q_{\theta}\left(\sum_{x}p^*(x)\alpha(x)\left(\frac{p_{\theta}(x)}{\sum_{\theta'}q_{\theta'}p_{\theta'}(x)}-1\right)\right),
\end{align*}
and $(\widehat{U}_n)_n$ the bounded martingale difference 
\begin{align*}
\widehat{U}_{n+1} :=\alpha(x_{n+1})B(q_n,x_{n+1}) - \mathbb{E}_{p^*}[\alpha(x_{n+1})B(q_n,x_{n+1})]. 
\end{align*}

The analysis is then identical to the one in Section \ref{sec:main}. In particular, if we denote  $p_{\alpha}^*(x): = \frac{p^*(x)\alpha(x)}{\sum_{x'}p^*(x')\alpha(x')}$ the \emph{as-if} data generating process, it is immediate to see that the predictive process converges almost surely to:
\begin{align*}
\arg\max_{p \in \co \mathcal{P}}\sum_{x \in X} p_{\alpha}^*(x) \log p(x).
\end{align*}
In the special case where $p_{\alpha}^* \in \co \mathcal{P}$, the predictive process converges to $p_{\alpha}^*$, which can be closer to $p^*$ than any of the $p_{\theta}$'s are.   

\begin{example}\label{ex:prior-bias} There are two states $A$ and $B$,  two possible observations $a$ and $b$, and $p_A(a) > p_B(a)$, that is, observing $a$ (resp., $b$) is more likely when the state is $A$ (resp., $B$). Assume that the agent initially believes that $A$ is more likely, i.e., $q_0(A) >1/2$ and let $\alpha(a) > \alpha(b)$. This simple setting captures a form of prior dependence, where the agent reacts more to information that confirms his initial belief. It is an extreme form of prior dependence, though, as the agent continues to react more to observation $a$, even  when his current belief $q_n$ indicates that $B$ is the most likely state. The \textit{as-if} data-generating process is:
\begin{align*}
p^*_{\alpha}(a)= \frac{p^*(a)\alpha(a)}{p^*(a)\alpha(a)+p^*(b)\alpha(b)}> p^*(a), 
\end{align*}
i.e., it is as-if the agent learns in an environment where the data-generating process is distorted towards $a$, the observation most indicative of $A$. For a concrete example, let $p^*(a)=2/3$, $p_A(a)=p_B(b)=3/4$, $\alpha(a) =3/4$ and  $\alpha(b)=1/4$, so that $p^*_{\alpha}(a) = 6/7$. The agent's predictive process converges to $p_A$, the same limiting predictive process a Bayesian agent would converge to.  If, however, the prior dependence is weaker, say $\alpha(a)= 2/3$ and $\alpha(b) =1/2$, then the predictive process converges to $(8/11, 3/11)$, which is strictly closer to the true data-generating process than $p_A$. \end{example}

So far, we have assumed that the weights are independent of the current belief $q_n$. Within the context of Example \ref{ex:prior-bias}, suppose instead that we write the updating weights as $\gamma_{n+1} \times\alpha(x_{n+1},q_n)$, with
$\alpha(a,q) = \alpha^+ =\alpha(b,1-q)$ when $q>1/2$, $\alpha(a,1-q) = \alpha^- =\alpha(b, q)$ when $q <1/2$, and $1 \geq \alpha^{+} > \alpha^{-} \geq 0$. When $q =1/2$, let $\alpha(a,1/2)=\alpha(b,1/2) =\alpha$. Thus, if the agent observes $a$ (resp. $b$) when he currently believes $A$ (resp. $B$) to be the most likely state, he revises his belief towards the Bayesian posterior more than when he observes $b$ (resp., $a$). This is a weaker form of ``prior'' dependence: The agent reacts more to information at period $n+1$ that confirms the belief $q_n$ they had at the end of period $n$. This updating process continues to be a Robbins-Monro algorithm with the function $\widehat{H}$:
\begin{align*}
\widehat{H}(q) := 
\begin{cases}
 p^*(a)\left[\frac{p_A(a)q}{p_A(a)q+p_B(a)(1-q)}\right]\alpha^{+} + p^*(b)  \left[\frac{p_A(b)q}{p_A(b)q+p_B(b)(1-q)}\right]\alpha^{-}, & q >1/2, \\
 p^*(a)\left[\frac{p_A(a)}{p_A(a)+p_B(a)}\right]\alpha + p^*(b)  \left[\frac{p_A(b)}{p_A(b)+p_B(b)}\right]\alpha, & q =1/2, \\  
  p^*(a)\left[\frac{p_A(a)q}{p_A(a)q+p_B(a)(1-q)}\right]\alpha^{-} + p^*(b)  \left[\frac{p_A(b)q}{p_A(b)q+p_B(b)(1-q)}\right]\alpha^{+}, & q <1/2. \\
\end{cases}
\end{align*}

For a concrete example, let $p^*(a)=2/3$, $p_A(a)=p_B(b)=3/4$, $\alpha^+ =3/4$, $\alpha^-=1/4$, and $\alpha=1/2$. The function $\widehat{H}$ is then
\begin{align*}
\widehat{H}(q):= 
\begin{cases}
\frac{2}{3}\left[\frac{3q}{1+2q} -q \right]\frac{3}{4} +   \frac{1}{3}\left[\frac{q}{3-2q}-q \right]\frac{1}{4} & q>1/2, \\
\frac{1}{24}  & q=1/2, \\
\frac{2}{3}\left[\frac{3q}{1+2q} -q \right]\frac{1}{4} +   \frac{1}{3}\left[\frac{q}{3-2q}-q \right]\frac{3}{4} & q<1/2.
\end{cases}
\end{align*}

We plot the function $\widehat{H}$ in Figure \ref{fig:self-conf}.  It is immediate to verify that $\widehat{H}$ has three  zeroes: $q=0$, $q=3/10$, and $q=1$, the latter two being locally stable. 

\begin{figure}[h!]
\centering
\includegraphics{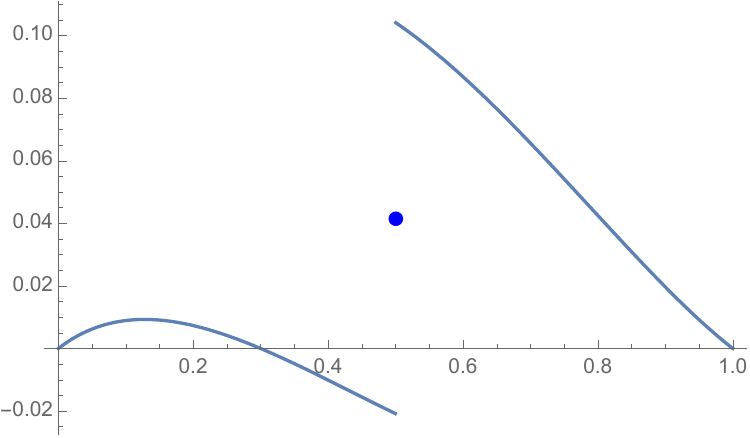}
\caption{The function $\widehat{H}$}\label{fig:self-conf}
\end{figure}

\medskip

The local stability of $q=3/10$ and $q=1$ is not surprising. Since $\alpha(x,\cdot)$ is constant in $q$ for all $q <1/2$, the cross-entropy 
\begin{align*}
\sum_{x \in \{a,b\}}\frac{p^*(x)\alpha(x,0)}{\sum_{x' \{a,b\}}p^*(x')\alpha(x',0)}\log\left(q p_A(x) + (1-q) p_B(x)\right)
\end{align*}
is maximized with respect to $q \in [0,1/2)$ in a stable zero of $\widehat{H}$, restricted to that sub-domain. It is immediate to verify that $q=3/10$ is indeed the maximizer. A symmetric argument applies for all  $q>1/2$.\medskip 

To sum up, our main insight, that  some non-Bayesian updating can outperform Bayesian updating in misspecified environments, is not limited to the setting in Section \ref{sec:main}.

\subsection{Generalized Bayes Rule}
A related updating rule is the generalized Bayes rule, where the revised belief $q_{n+1}(\theta)$ is proportional to 
\begin{align*}
q_n(\theta) p_{\theta}(x_{n+1})^{\gamma_{n+1}},
\end{align*}
when the agent observes $x_{n+1}$, $\gamma_{n+1} \in [0,1]$. The generalized Bayes rule is a special case of the \cite{grether} rule, 
captures some form of underreaction to new information and corresponds to Bayes rule when $\gamma_{n+1}=1$.  Recent contributions in statistics, e.g., \cite{grunwald2017} and \cite{grunwald2020}, advocate the use of the generalized Bayes rule, particularly as a method to correct for inconsistencies in Bayesian inference problems. 
\medskip 

We now argue that its long-run properties are markedly different from the rule (\ref{C-Bay}). To ease the discussion, assume that there exists a unique $p_{\theta^*}$, which maximizes the cross-entropy $V$, that is, $V(p_{\theta^*}) > V(p_{\theta})$ for all $\theta \neq \theta^*$. When the observations up to period $n+1$ are $(x_1,\dots,x_{n+1})$, we have that
\begin{align*}
\log \left(\frac{q_{n+1}(\theta^*)}{q_{n+1}(\theta)}\right) & = \log \left(\frac{q_{0}(\theta^*)}{q_{0}(\theta)}\right)  + \sum_{i=1}^{n+1}  \gamma_i \log \left(\frac{p_{\theta^*}(x_i)}{p_{\theta}(x_i)}\right), \\
& \geq \log \left(\frac{q_{0}(\theta^*)}{q_{0}(\theta)}\right)  + \\
& n \Big[ \Bigl(\max_{i=1,\dots,n+1}\gamma_i\Bigr) \frac{1}{n+1}\sum_{i=1}^{n+1}\log p_{\theta^*}(x_i)-  \Bigl(\min_{i=1,\dots,n+1} \gamma_i\Bigr)\frac{1}{n+1}\sum_{i=1,\dots,n+1}^{n+1}\log p_{\theta}(x_i)\Big].
\end{align*}
The strong law of large numbers implies that $\lim_{n \rightarrow + \infty} \frac{1}{n+1}\sum_{i=1}^{n+1}\log p_{\hat{\theta}}(x_i) = V(p_{\hat{\theta}})$ almost surely for all $\hat{\theta}$, hence the term into bracket converges to a strictly positive number since $V(p_{\theta^*}) >V(p_{\theta})$. Therefore, $\left(\log \left(\frac{q_{n+1}(\theta^*)}{q_{n+1}(\theta)}\right)\right)_n$ converges  to $+\infty$ and, consequently, $(q_{n+1}(\theta^*))_{n}$ converges to one almost surely. The generalized Bayes rule thus converges to the same degenerated belief as Bayes rule, in sharp contrast with the updating rule (\ref{C-Bay}).\medskip 

Thus, despite the similarities of the two rules, their asymptotic properties are fundamentally different. Intuitively, this is because the two rules underreact to new information on two different scales (linear vs. logarithmic), with the updating rule (\ref{C-Bay}) underreacting more and, consequently, protecting more against model's misspecification.

\subsection{Overreaction}

We have assumed that the agent underreacts to news, i.e., $\gamma_n < 1$. If the agent overreacts to news, the convergence to a limit point or cycle is not guaranteed in general. Intuitively, overreaction may lead to big swings in beliefs, which prevents them from stabilizing. For an illustration, see Figure \ref{fig:overreaction}. Clearly, the belief process (the green curve) is quite chaotic, with the beliefs varying wildly.\footnote{To draw Figure \ref{fig:overreaction}, we have assumed that $X = \{0,1\}$, $\Theta=\{\theta,\theta'\}$, $p^*(1)=2/3$, $p_{\theta}(1)=9/10 = p_{\theta'}(0)$, and $\gamma_n = 2 +1/n$. In instances where $(1-\gamma_n)q_n + \gamma_n B(x_{n+1},q_n) > 1$ (resp., $<0$), we let $q_{n+1} = 0.99$ (resp., $0.01$).} There is therefore little hope for general statements. We conjecture, however, that if $(\gamma_n)_n$ converges to 1, then we recover the convergence to the closest distribution in $\mathcal{P}$ to $p^*$.

\begin{figure}[h!]
\centering
\includegraphics[scale=0.7]{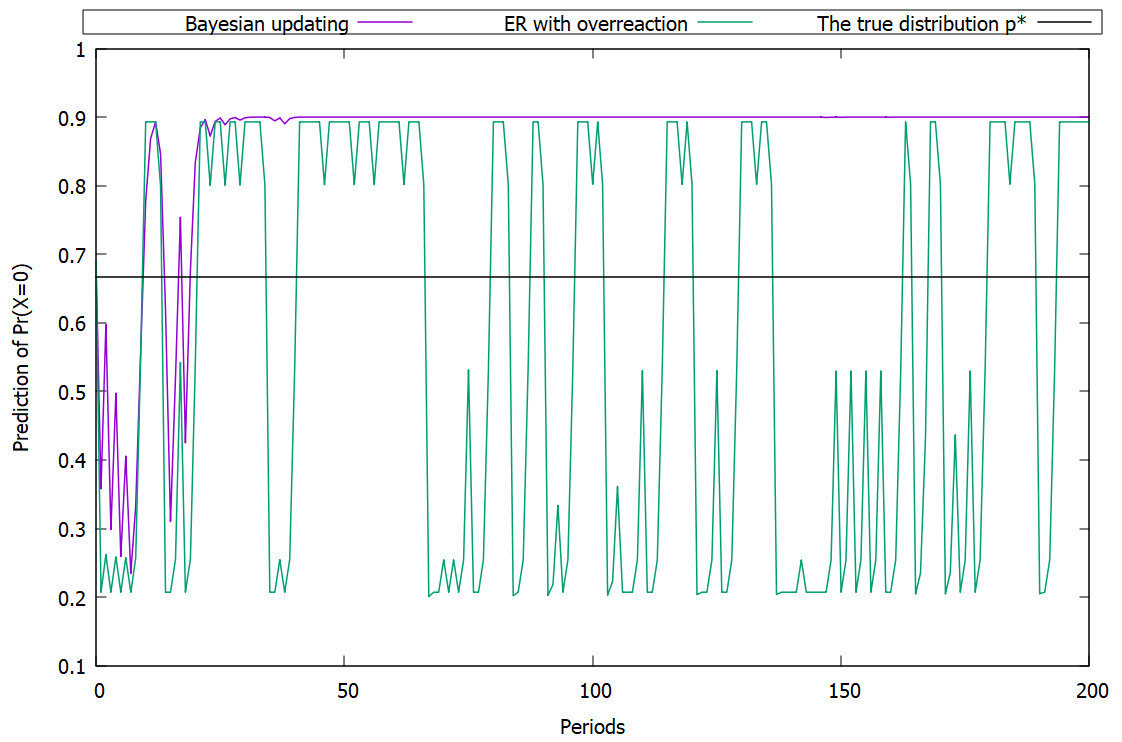}
\caption{Overreaction and Oscillation}\label{fig:overreaction}
\end{figure}

\section{Literature}\label{sec:lit}

Two closely related papers are \cite{heidhues2021} and \cite{frick2023belief}.  \cite{heidhues2021} consider the following model. In each period $n$, the agent makes a decision $a_n$ and observes a signal $q_n$, which depends on the decision $a_n$ and an (unobservable) external factor $x_n$. The authors assume that $x_n =\theta +\epsilon_n$, where $\theta$ is a fixed parameter and $\epsilon_n$ a realization of a  normal distribution with mean zero and precision $\rho_{\epsilon}$. (The realizations are independent.)  When the agent makes the decision $a_n$ and the external factor is $x_n$, the signal is $q_n=Q(a_n,x_n)$, with $Q$ strictly monotonic in $x$ for all $a$. The agent, however, perceives the relationship between decisions, external factors and signals to be $\widehat{Q}$, also strictly monotonic in $x$ for all $a$. The agent is Bayesian and initially believes that the unknown mean $\theta$ is normally distributed with  mean zero and precision $\rho_0$. Thus, upon observing $q_n$ after having chosen $a_n$, the agent thinks that the realized external factor is $\widehat{x}_n$, the unique solution to $\widehat{Q}(a_n, x)=q_n$, which may differ from the truly realized $x_n$. At period $n$, the agent's expectation $\widehat{\theta}_n$ of the mean therefore satisfies: 
\begin{align*}
\widehat{\theta}_n = \widehat{\theta}_{n-1} + \gamma_n [\widehat{x}_n-\widehat{\theta}_{n-1}], \text{where\;} \gamma_n=\frac{
\rho_{\epsilon}}{n \rho_{\epsilon} + \rho_{0}}.
\end{align*}
\cite{heidhues2021} rewrite the evolution of $\widehat{\theta}_n$ as a Robbins-Monro algorithm and prove almost sure convergence. (This discussion sidesteps an important issue, namely that, $\widehat{x}_n$ depends on the decision made $a_n$.) While \cite{heidhues2021} frame their problem as an active learning problem, they also explain how we can reframe it as a learning problem with non-Bayesian updating. For instance, the model accommodates underreaction for specifications of $Q$ and $\widehat{Q}$ such that $x_n > \widehat{x}_n >0$ and $x_n < \widehat{x}_n<0$. The model can also accommodate base-rate neglect and the confirmatory bias. If we interpret any deviation from $x_n$ as a departure from Bayesian updating, as the authors do, then Bayesian updating corresponds to $x_n=\widehat{x}_n$,  that is, for all decisions made, the agent's model is correct  (and, thus, learns $\theta$).  Unlike our framework, their framework is not flexible enough to cleanly decouple the updating rule from the agent's misspecification.  We stress that non-Bayesian updating in misspecified models was not their focus.  We differ in two other important ways. First, we do not restrict attention to  a particular family of distributions, such as the normal distributions. Second, we not only prove that the belief process converges almost surely, but also fully characterize the limit set.  \cite{heidhues2021} only prove convergence. In common, we use the ODE method to characterize the limit set of our stochastic approximation algorithm. See \cite{BerBraFau20}, \cite{esponda2021asymptotic}, \cite{gagnon2022learning} and \cite{danenberg2024} for some other recent applications of stochastic approximation in Economics and \cite{marcet1989} for an earlier application.  

\medskip

\cite{frick2023belief} propose a general model of Bayesian learning, where the true and perceived data generating processes $p_{q}^{*} \in \Delta(X)$ and $p_{\theta,\, q} \in \Delta(X)$, $\theta \in \Theta$, may depend on the current belief $q \in \mathbf{S}$.  The model is flexible enough to accommodate some departures from Bayesian updating as well as several models of active and social learning (in the latter case the dependence on the current beliefs is through the actions chosen). In particular,  as we have already argued,  the updating rule (\ref{C-Bay}) induces the same posteriors than the Bayesian updating of the ``as-if'' data generating processes  $p^{\gamma}_{\theta,q}\in \Delta(X)$, where   $p^{\gamma}_{\theta,q}(x):= (1-\gamma)(\sum_{\theta'}q_{\theta'}p_{\theta'}(x)) + \gamma p_{\theta}(x)$ for all $(x,\theta)$, when the current belief is $q$. The results of \cite{frick2023belief} are not applicable to our problem, however. First, their model cannot capture weights that depend on the number $n$ of observations, as we do. Second, and most importantly, their results provide conditions for the (in)stability of degenerated beliefs $\delta_{\theta}$, $\theta \in \Theta$.  None of their results addresses the stability properties of interior beliefs.  For other related papers on learning in misspecified models, see \cite{esponda2016berk}, \cite{fudenberg2017active}, \cite{heidhues2018unrealistic}, \cite{frick2020misinterpreting}, \cite{fudenberg2021limit}, \cite{bohren2021learning}, and \cite{lanzani2022}. 

 \medskip 
  
Relatedly, \cite{frick2024welfare} compare misspecified models in the spirit of \cite{blackwell51}. More precisely, consider two statistical models 
$\mathcal{P}^1:=(p_{\theta}^1)_{\theta}$ and $\mathcal{P}^2:=(p_{\theta}^2)_{\theta}$ and assume that the true model is $\mathcal{P}^*:=(p^*_{\theta})_{\theta}$.  The idea is to compare the expected payoff the agent would obtain if his decisions are based on either model. 
Formally, let $\mathcal{F}$ be a non-empty set of payoff acts $f: \Theta \rightarrow \mathbb{R}$ and $f_{\mathcal{F}}^i(x_1,\dots,x_n)$ a solution to $\max_{f \in \mathcal{F}}\mathbb{E}[f | (x_1,\dots,x_n), \mathcal{P}^i]$, that is, $f_{\mathcal{F}}^i(x_1,\dots,x_n)$ maximizes the agent's expected payoff given the observations $(x_1,\dots,x_n)$ and the model $\mathcal{P}^i$. (The authors assume Bayesian updating, but as already explained, some choices of $\mathcal{P}^i$ are equivalent to non-Bayesian updating.)  The model $\mathcal{P}^1$ outperforms the model $\mathcal{P}^2$ if there exists $n^*$ such that for all $n \geq n^*$, $\mathbb{E}_{\mathcal{P}^*}[f_{\mathcal{F}}^1(x_1,\dots,x_n)] \geq \mathbb{E}_{\mathcal{P}^*}[f_{\mathcal{F}}^2(x_1,\dots,x_n)]$ for all decision problems $\mathcal{F}$. \cite{frick2024welfare} provide tight conditions for one model to outperform the other.  There are no difficulties to adapt their definitions to our framework. It is, however, unclear how to adapt their results, since the updating rule (\ref{C-Bay}) necessitates a model, which changes with the current belief and the number of observations. Intuitively, though,  when $p^* \in \co \mathcal{P}$, the updating rule (\ref{C-Bay}) converges  to the true data-generating process, while Bayesian updating does not. For large enough $n$, the updating rule (\ref{C-Bay}) therefore outperforms Bayesian updating, in the sense of \cite{frick2024welfare}. Even when $p^* \notin \co \mathcal{P}$, the updating rule (\ref{C-Bay}) comes closer to the true data-generating process, hence there are decision problems (e.g., scoring problems), where it again outperforms Bayesian updating. \medskip

Finally, we have established that the updating rule (\ref{C-Bay})  induces (random) sequences of beliefs converging to solutions of 
\begin{align*}
\max_{q \in \Delta(\Theta)}\sum_{x \in X}p^*(x) \ln\left(\sum_{\theta}q_{\theta}p_{\theta}(x)\right). \tag{$\mathcal{M}$}
\end{align*}
This maximization problem ($\mathcal{M}$) has a well-known antecedent in Information Theory and Finance, where it is known as the \cite{kelly1956} problem.\footnote{We thank Jakub Steiner and Larry Samuelson for drawing our attention to this literature.} In Kelly model, there is an investor with initial wealth $W_0$, $|\Theta|$ assets, and $|X|$ states, with $p^*$ the distribution over states. States are i.i.d drawn. Asset $\theta$ returns $p_{\theta}(x)$, when the state is $x$. At each period, the investor chooses a portfolio $q$, with $q_{\theta}$ the fraction of the current wealth allocated to asset $\theta$. Kelly assumes that the investor's objective is to maximize the long-term growth rate of wealth and shows that there exists an optimal stationary strategy that consists in choosing the portfolio $q^*$ at each period, where $q^*$ is a solution to ($\mathcal{M}$). For a detailed exposition, see \cite{cover-thomas}. 

More closely related to our work, \cite{cover1984} and \cite{cover-gluss1986} develop algorithms for solving ($\mathcal{M}$). Neither, however, coincides with the updating rule (\ref{C-Bay}). \cite{cover1984} assumes that the distribution $p^*$ is known and proposes a gradient-ascent procedure. The procedure depends on $p^*$ and, thus, differs from the procedure (\ref{C-Bay}). Like us,  \cite{cover-gluss1986} assume that $p^*$ is unknown. Their approach, however, differs substantially from ours. Roughly, their algorithm computes sequences of solutions to ($\mathcal{M}$), substituting $p^*$ for the corresponding empirical distributions.

\section{Proofs}
\subsection{Proof of Theorem \ref{th:limitset}}
We equip the  space $\mathbf{S}$ with the metric $d(q,q') := \max_{\theta \in \Theta} |q_{\theta} - q'_{\theta}|$.\medskip

The proof uses methods and results from stochastic approximation theory. We refer the reader to \cite{Ben99} for a short exposition and to \cite{borkar2008} and \cite{kushner-yin-2003} for book-length treatments. \medskip 

The roadmap is as follows. We first rewrite the updating rule (\ref{C-Bay}) as a Robbins-Monro algorithm, with associated ODE $\dot{q}(t) = H(q(t))$. We then analyze the associated ODE. In particular, we show that the cross-entropy map $V$ is a Lyapunov function for the associated flow and that the set $E$ can be decomposed into finitely many closed convex (hence, connected) components (the decomposition is unique). By Proposition 6.4 in \cite{Ben99}, the (random) limit set of $(q_n)_n$ is included in a closed connected subset of $E$. Finally, we show that the probability to converge to the component that maximizes the cross-entropy is one. 

\medskip 

We rewrite the updating rule (\ref{C-Bay}) as: 
\begin{align*}
q_{n+1} & = q_n + \gamma_{n+1}\left[ B (q_n,x_{n+1}) -q_n\right] \\
& = q_n + \gamma_{n+1}\left[ H(q_n) + U_{n+1}\right],
\end{align*}
where $(U_n)_n$ is the bounded martingale difference with $U_{n+1} :=B(q_n,x_{n+1}) - \mathbb{E}_{p^*}[B(q_n,\mathbf{x}_{n+1}]$ when the realized signal is $x_{n+1}$, and $H(q_n):= \mathbb{E}_{p^*}[B(q_n,\mathbf{x}_{n+1})] -q_n$.

\noindent The ODE associated with the random process $(q_n)_n$  is
\begin{equation} \label{eq:dynamics}
\dot{q}(t) = H(q(t)).
\end{equation}
For further reference, note that the component-wise system is given by
\begin{equation}
\dot{q}_{\theta}(t) =   q_{\theta}(t) (-1 + f_{\theta}(q(t))), \; \, \text{with} \; f_{\theta}(q):=\sum_{x}p^*(x)\frac{p_{\theta}(x)}{\sum_{\theta'}q_{\theta'}p_{\theta'}(x)},  \forall \theta \in \Theta.
\end{equation}

We first show that the cross-entropy map $V$ grows along the trajectory of the dynamics (\ref{eq:dynamics}). 

\begin{lemma}\label{lem:V-Lyapunov}
    $V$ is a strict Lyapunov function for the flow $\phi_t$ associated to (\ref{eq:dynamics}),
    meaning that, along any non-stationary solution curve $(q(t))_{t \geq 0}$ we have $\frac{d}{dt} V(q(t)) > 0$. 
\end{lemma}

\begin{proof}[ of Lemma \ref{lem:V-Lyapunov}] Note that 
\[ \frac{d}{dt} V(q(t)) = \left<\nabla V (q(t)) , \dot{q}(t)  \right> =  G(q(t))\] for all $t \geq 0$, where $G(q) := \sum_{\theta} \frac{\partial V}{\partial q_{\theta}}(q) q_{\theta} (-1 + f_{\theta}(q))$. 

Since \[\frac{\partial V}{\partial q_{\theta}}(q) = \sum_{x \in X} p^*(x) \frac{p_{\theta}(x)}{\sum_{\theta' \in \Theta} p_{\theta'}(x) q_{\theta'}} = f_{\theta}(q),\] we have \[G(q) = \sum_{\theta \in \Theta} q_{\theta} f_{\theta}(q) (-1 + f_{\theta}(q)) = \sum_{\theta \in \Theta} q_{\theta} \left(-1 + f_{\theta}(q) \right)^2 + \sum_{\theta \in \Theta} q_{\theta} f_{\theta}(q) -\sum_{\theta \in \Theta} q_{\theta}.\]

Notice that $\sum_{\theta \in \Theta} q_{\theta} f_{\theta}(q) = 1$, and since $\sum_{\theta \in \Theta} q_{\theta} = 1$, we have $G(q) \geq 0$, with equality if and only if $f_{\theta}(q) = 1$ for all $\theta$ such that $q_{\theta}>0$. Hence, $G(q) \geq 0$ with equality if and only if $q \in E$, which concludes the proof.
 \end{proof}

Lemma \ref{lem:V-Lyapunov} implies that the only invariant sets for the flow are connected components of $E$, the set of zeroes of $H$. By Proposition 6.4 in \cite{Ben99}, it also implies that $\mathcal{L}(q_n)$, the (random) limit set of $(q_n)_n$, is included in a closed connected subset of $E$.\footnote{The decomposition of a set into its connected components is always unique.} We now characterize $E$.

\begin{lemma} \label{lem:decomposition}

\begin{itemize}
 \item[(i)] The set $E$
      is a finite union of disjoint non-empty compact and convex components:  there exists $K < +\infty$ such that $E = \bigcup_{k=1}^{K}C_k$, where each $C_k$ is a non-empty compact convex subset of $\mathbf{S}$ and $C_k \cap C_{k'} = \emptyset$ for all $(k,k')$ with $k \neq k'$.    
\item[(ii)] The cross-entropy $V$ takes the constant value $V(C_k)$ on each component $C_k$, and there exists a unique $k^*$ such that $V(C_{k^*}) > V(q)$ for all $q \notin C_{k^*}$.
\end{itemize}
   \end{lemma}

\begin{proof}[ of Lemma \ref{lem:decomposition}] For any non-empty subset $\widehat{\Theta}$ of $\Theta$, we write $F_{\widehat{\Theta}}$ for the subset of $\mathbf{S}$:
\[F_{\widehat{\Theta}} := \left\{q \in \mathbf{S}: \; \, q_{\theta}= 0 \text{ iff } \; \theta \notin \widehat{\Theta} \right\}.\]
In particular,  $F_{\Theta} = \mathbf{S}^*$ and $F_{\{\theta\}} = \{\delta_{\theta}\}$. \medskip 

Recall that $q \in E$ if  and only if 
\[q_{\theta}(f_{\theta}(q)-1)=0, \]
for all $\theta$. Consider the constrained maximization problem $\max_{q \in F_{\widehat{\Theta}}}V(q)$, that is, 
\[\max_{q \in \mathbb{R}^{|\Theta|}}V(q),\] 
subject to $\sum_{\theta \in \widehat{\Theta}}q_{\theta}=1$, $q_{\theta} > 0$ for all $\theta \in \widehat{\Theta}$ and $q_{\theta}=0$ for all $\theta \notin \widehat{\Theta}$. If a solution exists, for all $\theta \in \widehat{\Theta}$, the first-order condition with respect $q_{\theta}$ satisfies
\begin{align*}
f_{\theta}(q) := \sum_x p^*(x) \frac{p_{\theta}(x)}{\sum_{\theta'}q_{\theta'}p_{\theta'}(x)} = \lambda,
\end{align*}
where $\lambda$ is the Kuhn-Tucker multiplier associated with $\sum_{\theta \in \widehat{\Theta}}q_{\theta}=1$. (The complementary slackness condition implies that the multipliers associated with $q_{\theta} > 0$ for all $\theta \in \widehat{\Theta}$ are all zero.)  Multiplying each equation by $q_{\theta}$ and summing over $\theta \in \widehat{\Theta}$, we get 
\[\sum_{\theta \in \Theta} q_{\theta} f_{\theta}(q) = \lambda \sum_{\theta \in \Theta} q_{\theta} = \lambda. \] Since the left hand-side is equal to $1$, it follows that \[f_{\theta}(q)=1 \text{\;\; for all } \theta \in \widehat{\Theta}.\]

Therefore, at a maximum, we have 
\[q_{\theta}(f_{\theta}(q)-1) = 0, \]
for all $\theta \in \Theta$, i.e., it is a zero of $H$. Since $V$ is concave, the converse is also true by the Kuhn-Tucker-Karush theorem. It follows that $E$ is the union of  the sets  $S_{\widehat{\Theta}}:=\arg\max_{q \in F_{\widehat{\Theta}}}V(q)$ over all possible non-empty subsets $\widehat{\Theta}$ of $\Theta$.  \medskip

Since $V$ is concave, $S_{\widehat{\Theta}}$ is convex for all $\widehat{\Theta}$. We write $V(S_{\widehat{\Theta}})$ for the value taken by $V$ on $S_{\widehat{\Theta}}$. We now construct the decomposition from the sets $S_{\widehat{\Theta}}$. 
\medskip

Let $T := \left\{\widehat{\Theta} \subseteq \Theta: \; \, S_{\widehat{\Theta}} \neq \emptyset  \right\}$. Given $(\widehat{\Theta}, \widehat{\Theta}') \in T \times T$ say that $\widehat{\Theta} \sim \widehat{\Theta}'$ if there exists a family $\{\widehat{\Theta}_k\}_{k=0,\dots,K}$ such that $\widehat{\Theta}=\widehat{\Theta}_0$, $\widehat{\Theta}'=\widehat{\Theta}_K$, and $\mathrm{dist} \left(S_{\widehat{\Theta}_k}, S_{\widehat{\Theta}_{k+1}}\right) = 0$ for $k=0,...,K-1$.\footnote{The distance is the Hausdorff distance between subsets of $\mathbf{S}$, with $\mathbf{S}$ equipped with the total variation metric.} The binary relation $\sim$ is an equivalence relation. Let $(\mathcal{C}_k)_{k=1,...,K}$ be its equivalence classes  (note that $K \leq 2^{|\Theta|} -1$) and  $C_k := \bigcup_{\widehat{\Theta} \in \mathcal{C}_k} S_{\widehat{\Theta}}$. To conclude the proof, we need to show that $C_k$ is closed and convex, and that the global maximum of $V$ is attained on only one of these sets.\medskip

We first show that $C_k$ is closed: let $(q^i)_i$ be a sequence in $C_k$ and suppose that $\lim_{i \rightarrow + \infty} q^i = q$, with $q \in F_{\widehat{\Theta}}$. By continuity, $V(q) = V(C_k)$. Moreover, for $i$ large enough, $q^i_{\theta} >0$, for all $\theta \in \widehat{\Theta}$. Assume by contradiction that $q \notin S_{\widehat{\Theta}}$, then there exists $\tilde{q} \in F_{\widehat{\Theta}}$ such that $V(\tilde{q}) > V(q)$. There exists a sequence $\tilde{q}^i$ such that $\tilde{q}^i$ is in the same face as $q^i$, and $\lim_{i \rightarrow + \infty} \tilde{q}^i = \tilde{q}$. Therefore
\[\lim_{i \rightarrow + \infty} V(q^i) =  V(q) < V(\tilde{q}) = \lim_{i \rightarrow + \infty} V(\tilde{q}^i), \]
implying that, for large enough $i$, $V(\tilde{q}^i)> V(q^i)$. This contradicts the fact that $q^i$ maximizes $V$ on the face it belongs to. Therefore, $q \in S_{\widehat{\Theta}}$ for some $\widehat{\Theta} \in \mathcal{C}_k$.
Consequently, the collection $\{C_k\}$ is composed of  closed, connected,  pairwise disjoint sets, and 
\begin{equation} \label{eq:decomp}
E = \bigcup_{k=1}^K C_k.
 \end{equation}

We now argue that $C_k$ is convex, for $k=1,...,K$. 
\noindent Define $W: \Delta(X) \rightarrow \mathbb{R}$ as follows:
\[p \in \Delta(X) \mapsto  W(p) := \sum_x p^*(x) \log p(x).\] 
Since $p^*$ has full support on $X$, $W$ is strictly concave on $\Delta(X)$.\footnote{More generally, if the support of $p^*$ is not $X$, the argument remains valid if we restrict  $W$ to the support of $p^*$.} Let $L: \mathbf{S} \rightarrow \Delta(X)$ be the linear map $q \mapsto   \left(\sum_{\theta} q_{\theta} p_{\theta}(x)\right)_{x \in X}$. By linearity, $L(S_{\widehat{\Theta}})$ is a convex subset of $\Delta(X)$. By strict concavity of $W$, $W$ is constant on $L(S_{\widehat{\Theta}})$ if, and only if, $L(S_{\widehat{\Theta}})$ is a singleton. Since $V$ is constant on $S_{\widehat{\Theta}}$, and $V = W \circ L$, this implies that $L(S_{\widehat{\Theta}})$ is a singleton.  
Call $p_{\widehat{\Theta}} := \{L(S_{\widehat{\Theta}}) \} \in \Delta(X)$. By continuity of $L$, $p_{\widehat{\Theta}} = p_{\widehat{\Theta}'}$ for $(\widehat{\Theta}, \widehat{\Theta}') \in \mathcal{C}_k \times \mathcal{C}_k$, so that $L$ is constant on $C_k$ (notice that this is stronger than stating that $V$ is constant along $C_k$). Let $(q,q')\in C_k \times C_k$, with $q \in S_{\widehat{\Theta}}, q' \in S_{\widehat{\Theta}'}$, and $\lambda \in (0,1)$. Then, if $q'' := \lambda q + (1-\lambda) q'$, we have 
$L(q'') = L\left(\lambda q + (1-\lambda)q'\right) = \lambda L(q) + (1-\lambda) L(q') = L(C_k)$.
Since  $q$ and $q'$ are zeroes of  $H$, we have
\[1 = f_{\theta}(q) = \sum_x p^*(x) \frac{p_{\theta}(x)}{L(C_k)(x)}, \; \forall \theta \in \widehat{\Theta}, \; \text{ and } \; \, 1 = f_{\theta}(q) = \sum_x p^*(x) \frac{p_{\theta}(x)}{L(C_k)(x)}, \; \forall \theta \in \widehat{\Theta}'.\]
Now, let $\theta$ be such that $q''_{\theta} \neq 0$. then either  $q_{\theta}  \neq  0$ or $q'_{\theta}\neq 0$. In other terms, $\theta \in \widehat{\Theta} \cup \widehat{\Theta}'$.  
Since $\lambda \in (0,1)$, $q''_{\theta}>0$ if and only if $\theta \in \widehat{\Theta} \cup \widehat{\Theta'}$. Moreover, 
\[f_{\theta}(q'') = \sum_x p^*(x) \frac{p_{\theta}(x)}{L(q'')(x)} = \sum_x p^*(x) \frac{p_{\theta}(x)}{L(C_k)(x)} = 1,\]
for all $\theta \in  \widehat{\Theta} \cup \widehat{\Theta'}$. Hence, $q''$ is a zero of $H$ in $S_{\widehat{\Theta} \cup \widehat{\Theta}'}$ and, consequently, $\co(C_k) \subseteq E$. It follows that $\cup_{k'} C_{k'} =E = \cup_{k'} \co(C_{k'})$.  Now, if $C_k \neq \co(C_k)$, then $\co(C_k) \subseteq \cup_{k' \neq k} C_{k'}$, a contradiction with the fact that the components are pairwise disjoint. Therefore, $C_k = \co(C_k)$, that is, $C_k$ is convex.   \medskip

Finally, since $\arg\max_{q \in \mathbf{S}} V(q)$ is a convex set, there exists a unique index $k^* \in \{1,...,K\}$ such that $C_{k^*} = \arg\max_{q \in \mathbf{S}} V(q)$. This concludes the proof. 
 \end{proof}

 The final step consists in proving that the convergence to any $C_k$ with $k \neq k^*$ has probability zero. The following observations play a crucial role. 
Since $C_k$ is convex (and not merely connected), 
the set $\mathcal{C}_k$ admits a maximal element, that is, there exists  $\widehat{\Theta}(k) \in \mathcal{C}_k$ such that $\widehat{\Theta} \subset  \widehat{\Theta}(k)$, for all $\widehat{\Theta} \in  \mathcal{C}_k \setminus \{\widehat{\Theta}(k)\}$. We call $\widehat{\Theta}(k)$ the \emph{support} of $C_k$. By construction, for $k \neq k^*$,  $\widehat{\Theta}(k) \neq \Theta$. We also have: 

\begin{lemma} \label{lem:supports}
For all $k \neq k^*$, the set $\widehat{\Theta}(k^*) \setminus \widehat{\Theta}(k)$ is non-empty.
 \end{lemma}

\begin{proof}[ of Lemma \ref{lem:supports}] If it were not the case, we would have $\widehat{\Theta}(k^*) \subseteq \widehat{\Theta}(k)$. Therefore, $\max_{q \in \Delta(\widehat{\Theta}(k^*))} V \leq \max_{q \in \Delta(\widehat{\Theta}(k))
} V$ and, thus, $V(C_{k^*}) \leq V(C_k)$, a contradiction. 
\end{proof}

We are now ready to state and prove the final step in the proof. 

\begin{lemma}\label{lem:C_k}
For all $k \neq k^*$, $\mathbb{P} \left(\mathcal{L}(q_n) \subseteq C_k \right) = 0.$ 
\end{lemma}

\begin{proof}[ of Lemma \ref{lem:C_k}]
Recall that $C_k = \bigcup_{\widehat{\Theta} \in \mathcal{C}_k} S_{\widehat{\Theta}}$. Pick  $q^* \in F_{\widehat{\Theta}(k^*)}$, and let
\[\Theta_3 := \bigcup_{q \in C_k}\left\{\theta \in \Theta: \; q^*_{\theta}>0, \; q_{\theta} = 0\right\}.\]
For $\hat{q} \in C_k$, define $\Theta_1(\hat{q}):= \left\{\theta \in \Theta: \; q^*_{\theta} = \hat{q}_{\theta} = 0 \right\}$, and $\Theta_2(\hat{q}) :=  \{\theta \notin \Theta_3: \; \hat{q}_{\theta} >0 \}$.\medskip 

By Lemma \ref{lem:supports}, there exists $\theta^* \in \Theta_3$ such that $\hat{q}_{\theta^*}=0$ for all $\hat{q} \in C_k$.   We stress that the sets $\Theta_1(\hat{q})$ and $\Theta_2(\hat{q})$ depend on the face $\hat{q}$ belongs to.
\medskip 

By strict concavity of $V$ on $\co \{q^*,\hat{q}\}$,  we have $\frac{\partial}{\partial t}V\left(t q^* + (1-t)\hat{q} \right) >0$. Hence,
\[0< \sum_{\theta \in \Theta} (q^*_{\theta} - \hat{q}_{\theta}) f_{\theta}(\hat{q}) =  \sum_{\theta \in \Theta} q^*_{\theta} f_{\theta}(\hat{q}) - 1 =  \sum_{\theta \in \Theta_2(\hat{q}) \cup \Theta_3} q^*_{\theta} f_{\theta}(\hat{q}) - 1  = \sum_{\theta \in \Theta_2(\hat{q})} q^*_{\theta} +  \sum_{\theta \in \Theta_3} q^*_{\theta}  f_{\theta}(\hat{q}) - 1.\]
Thus, $\sum_{\theta \in \Theta_3} q^*_{\theta}  f_{\theta}(\hat{q}) > 1 - \sum_{\theta \in \Theta_2(\hat{q})} q^*_{\theta}$. Since $\sum_{\theta \in \Theta_2(\hat{q})\cup \Theta_3} q^*_{\theta} = 1$, 
\[\sum_{\theta \in \Theta_3} \mu_{\theta}  f_{\theta}(\hat{q}) >1, \; \, \forall \hat{q} \in C_k, \]
where $\mu_{\theta} := \frac{q^*_{\theta}}{ \sum_{\theta \in \Theta_3} q_{\theta}^*}$, for all $\theta \in \Theta_3$.\medskip 

 By compactness of $C_k$, there exists $\varepsilon >0$ such that  $\sum_{\theta \in \Theta_3} \mu_{\theta} f_{\theta}(\hat{q}) >1 + 3 \varepsilon$, for all $\hat{q} \in C_k$. By continuity of $f_{ \theta}$, there exists an open neighborhood $U$ of $C_k$ (in $\mathbf{S}$) such that
\[\inf_{\hat{q} \in U} \sum_{\theta \in \Theta_3} \mu_{\theta} f_\theta(\hat{q}) > 1 + 2\varepsilon.\]
%Define $z_{\theta} := - \log q_{\theta}$. We have
%Define, for $q \in Int \left(\mathbf{S}\right)$, the quantity $\zeta(q) := \sum_{\theta \in \Theta_3} \mu_{\theta} \log q_{\theta}$. 
\medskip

Note that, if $q_{\theta} >0$, then  $0 \leq \frac{B_{\theta}(q,x)}{q_{\theta}} = \frac{p_{\theta}(x)}{\sum_{\theta'}q_{\theta'}p_{\theta'}(x)} \leq \frac{p_{\theta}(x)}{\min_{\theta' \in \Theta} p_{\theta'}(x)} < + \infty$ since $p_{\theta'}(x) >0$ for all $(\theta',x)$.\footnote{Again, if the support of $p^*$ was a strict subset the support of each $p_{\theta'}$, the analysis would hold for all $x$ in the support of $p^*$.} Let $C(\theta) := \max_{x \in X}  \frac{p_{\theta}(x)}{\min_{\theta' \in \Theta} p_{\theta'}(x)}$. If $q_{n,\theta} >0$, it follows that 
\begin{equation} \label{eq:bound1} 1 - \gamma_{n+1} + \gamma_{n+1} C(\theta) \geq   \frac{q_{n+1,\theta}}{q_{n,\theta}} =  1 - \gamma_{n+1} + \frac{\gamma_{n+1}}{q_{n,\theta}}B_{\theta}(q_n,x_{n+1}) \geq 1 - \gamma_{n+1}.\end{equation}
Since
\begin{equation} \label{eq:bound2} \frac{q_{n+1,\theta}}{q_{n,\theta}} =  1 + \gamma_{n+1} \left(-1 + f_{\theta}(q_n) + \frac{U_{n+1,\theta}}{q_{n,\theta}}\right). \end{equation}
we can combining (\ref{eq:bound1}) and (\ref{eq:bound2}) to obtain  $C(\theta) - 1 \geq -1 +   f_{\theta}(q_n) + \frac{U_{n+1,\theta}}{q_{n,\theta}} \geq -1$, for all $n$. \medskip  

From the Stochastic Taylor Theorem (see e.g. \cite{aliprantis-border}, Theorem 17.17 p. 569), for any $n \in \mathbb{N}^*$, there exists a random variable $\zeta_n$, such that $|\zeta_n| \leq  \gamma_{n+1} \left(-1 + f_{\theta}(q_n) + \frac{U_{n+1,\theta}}{q_{n,\theta}}\right)$, and
\[\log \left(\frac{q_{n+1,\theta}}{q_{n,\theta}} \right) = \gamma_{n+1} \left(-1 + f_{\theta}(q_n) + \frac{U_{n+1,\theta}}{q_{n,\theta}}\right) - \frac{1}{(1+ \zeta_n)^2} \gamma_{n+1}^2\left(-1 + f_{\theta}(q_n) + \frac{U_{n+1,\theta}}{q_{n,\theta}} \right)^2.\]
Defining $z_n :=  \sum_{\theta \in \Theta_3} \mu_{\theta} \log q_{n,\theta}$, it follows that 
\begin{eqnarray*} 
z_{n+1} - z_n &=& \sum_{\theta \in \Theta_3} \mu_{\theta} \log \left(\frac{q_{n+1,\theta}}{q_{n,\theta}} \right) \\
 &=&  \sum_{\theta \in \Theta_3} \mu_{\theta} \log \left(1 + \gamma_{n+1} \left( -1 +   f_{\theta}(q_n) + \frac{U_{n+1,\theta}}{q_{n,\theta}} \right) \right) \\
& =&  \gamma_{n+1} \sum_{\theta \in \Theta_3} \mu_{\theta}  \left(-1 +f_{\theta}(q_n) + \frac{U_{n+1,\theta}}{q_{n,\theta}} \right) - K_{n+1} \gamma_{n+1^2}, 
\end{eqnarray*}
where $K_{n+1} := \frac{1}{(1+ \zeta_n)^2} \left(-1 + f_{\theta}(q_n) + \frac{U_{n+1,\theta}}{q_{n,\theta}} \right)^2$. Note that for $n$ large enough $|K_{n+1} | \leq K := 2 \max \{1,(C(\theta)-1)^2\}$.  $|\zeta_n| \leq 1/2$.
Taking the conditional expectation, we get
\[\mathbb{E} \left(z_{n+1} \mid \mathcal{F}_n \right) - z_n = \gamma_{n+1 } \sum_{\theta \in \Theta_3} \mu_{\theta}    \left(-1 +f_{\theta}(q_n) \right) - \gamma_{n+1}^2 \mathbb{E} \left(K_{n+1} \mid \mathcal{F}_n\right)  \geq  \gamma_{n+1} \varepsilon - K \gamma^2_{n+1}, \] 
on the event $\{q_n \in U\}$, for large enough $n$. Assume by contradiction that $\mathbb{P}\left(\mathcal{L}(q_n) \subseteq U \right)>0$. Since 
\[\left\{\mathcal{L}(q_n) \subseteq U \right\} \subseteq \bigcup_{N \in \mathbb{N}} \left\{ q_n \in U, \;  \forall n \geq N \right\},\]
 there exists $\bar{N} \in \mathbb{N}$ such that 
\[ \mathbb{P} \left( q_n \in U, \;  \forall n \geq \bar{N} \right) >0.\]
However, on the event  $\left\{ q_n \in U, \;  \forall n \geq \bar{N} \right\}$,  we have $ \sum_{\theta \in \Theta_3} \mu_{\theta}    \left(-1 +f_{\theta}(q_n) \right) > \varepsilon$ for all $n \geq \bar{N}$, implying that, for any $n > \bar{N}$,
\[\mathbb{E}(z_{n+1}) \geq \mathbb{E}(z_{\bar{N}}) + \varepsilon \sum_{m=\bar{N}}^n \gamma_{n+1} - K \sum_{m=\bar{N}}^n \gamma_{n+1}^2. \]
Thus, on the event  $\left\{ q_n \in U, \;  \forall n \geq \bar{N} \right\}$, we have $\lim_{n \rightarrow + \infty} \mathbb{E}(z_n) = + \infty$, which cannot happen since $z_n \leq 0$. Therefore  $\mathbb{P} \left( q_n \in U, \;  \forall n \geq \bar{N} \right) = 0$, the desired contradiction.
\end{proof}

\subsection{Proof of Proposition \ref{prop:full}}
Let $\tilde{q} \in \mathbf{S}^*$ be a zero of $H$, that is, for all $\theta \in \Theta$, 
\[1 = f_{\theta}(\tilde{q}) = \sum_{x \in X} \frac{p^*(x)}{\sum_{\theta'} p_{\theta'}(x) \tilde{q}_{\theta'}} p_{\theta}(x).\]
Define $u \in \mathbb{R}^{|X|}$ as $u(x) :=  \frac{p^*(x)}{\sum_{\theta'} p_{\theta'}(x) \tilde{q}_{\theta'}}$. We have
\[\left< u , p_{\theta} - p_{\theta'}\right> = 0, \; \forall \theta, \theta'.\]
Since the family $\{p_{\theta}\}_{\theta \in \Theta}$ is full, this means that $u$ is orthogonal to all vectors in the tangent space $\{v \in \mathbb{R}^{|X|}: \sum_{x \in X}v_x=0\}$, that is, $u$ is proportional to the unit vector $(1, \dots,1)$.  Therefore, there exists a real number $\mu$ such that $p^*(x) = \mu 
\sum_{\theta'} p_{\theta'}(x) \tilde{q}_{\theta'}$, for all $x \in X$. Summing over $X$, we obtain that $\mu = 1$. Hence $p^* = \sum_{\theta \in \Theta} p_{\theta} \tilde{q}_{\theta}$, and $\tilde{q} \in \Lambda$, which completes the proof.

\subsection{Proof of Proposition \ref{prop:convex-indep}}

Let $\widehat{\Theta}$ be a non-empty subset of $\Theta$. The argmax of the map $p \in \co\{p_{\theta}\}_{\theta \in \widehat{\Theta}} \mapsto W(p)$ is a singleton, $\{\bar{p}\}$, by convexity of  $\co\{p_{\theta}\}_{\theta \in \widehat{\Theta}}$ and strict concavity of $W$. Since the family is free and $|\widehat{\Theta}| \leq |\Theta| \leq |X|$, the map $\mu: q  \mapsto \sum_{\theta \in \Theta}q_{\theta}p_{\theta} $ is injective from $F_{\widehat{\Theta}}$ to $\Delta(X)$. Therefore,   there exists a unique $\bar{q} \in F_{\widehat{\Theta}}$ such that $\mu(\bar{q}) = \bar{p}$.

\subsection{Proof of Theorem \ref{th:constant-step}} 
Recall that Theorem \ref{th:constant-step} characterizes the long-run properties of the updating process when the updating weights are constant, i.e., $\gamma_n = \gamma >0$ for all $n$. The proof borrows and adapts ideas from the recent work of \cite{BenSch19}.\medskip

We consider the family of random processes $\{(q_n^{\gamma})_n\}$, where $q_0^{\gamma}=q_0
\in \mathbf{S}^*$, and $(q_n^{\gamma})_n$ is recursively given by
\begin{equation} \label{eq:constant_step}
q^{\gamma}_{n+1} = (1-\gamma) q^{\gamma}_{n}  + \gamma B\left(q^{\gamma}_{n},x_{n+1}\right).
\end{equation}

 It will be convenient to rewrite the recursive formula (\ref{eq:constant_step}) component-wise as follows: for $\theta \in \Theta$, 
\begin{equation} \label{eq:ratio}
\frac{q^{\gamma}_{n+1, \theta}}{ q^{\gamma}_{n,\theta}} =  F_{\theta}^{\gamma} \left(q^{\gamma}_{n},x_{n+1}\right), \text{ where} \; F_{\theta}^{\gamma}(q,x) := (1-\gamma)  + \gamma \frac{p_{\theta}(x) }{\sum_{\theta'} p_{\theta'}(x) q_{\theta'}}.
\end{equation}
\vspace{.2cm}

\noindent 
%Given $\gamma >0$ the Markov process $(q^{\gamma}_n)_n$ can be written as $q^{\gamma}_{n+1} = F^{\gamma}_{x(n+1)}(q(n)^{\gamma})$. It is a so-called \emph{random dynamical system} (see e.g. \cite{MeyTwe12} or \cite{BenHur22}) 
Let $P_{\gamma}$ be the transition kernel associated to the Markov chain $(q^{\gamma}(n))_n$.
Given $h \in  \mathcal{C}\left(\mathbf{S}\right)$, the set of continuous maps from $\mathbf{S}$ to $\mathbb{R}$, define $P_{\gamma} h$ as
\[q \in \mathbf{S} \mapsto  P_{\gamma}h(q) := \mathbb{E}\left[h(q^{\gamma}_{1})  \mid q^{\gamma}_{0} = q \right] =  \int_{\mathbf{S}} h(\lambda) P_{\gamma}(q,d\lambda). \]

\vspace{.2cm}

\noindent Given $\gamma >0$, define the \emph{occupation measure up to time $n$} as 
\begin{equation}
\Pi^{\gamma}_n := \frac{1}{n} \sum_{m=0}^{n-1} \delta_{q^{\gamma}_{m}}, \, \text{ where } \; \delta_q \, \text{ denotes the Dirac measure at } \, q. 
\end{equation}
We are interested in  the asymptotic behavior of $(\Pi^{\gamma}_n)_n$, when $\gamma$ gets small. 

\noindent Following \cite{BenSch19}, define
\[r^{\gamma}_{\theta}(q) :=  \sum_{x \in X} p^*(x) \log F^{\gamma}_{\theta}(q,x) .\]
 Intuitively, this quantity represents the expected growth of the belief in $\theta$, when current belief is $q$ and observations are distributed according to $p^*$.

\begin{lemma} \label{lem:pos_rate} For all $\lambda^* \in C_{k^*}$ and $q \in \mathbf{S}$, the derivative of the mapping $\gamma \mapsto \sum_{\theta \in \Theta} \lambda^*_{\theta} r^{\gamma}_{\theta}(q)$ evaluated in zero is non-negative, and equal to zero if, and only if, $q \in C_{k^*}$. 

\noindent Consequently, for any $q \in \mathbf{S} \setminus C_{k^*}$, there exists $\gamma_0(q)  >0$ such that, for all $0< \gamma < \gamma_0(q)$ and  all $\lambda^* \in C_{k^*}$, we have  $\sum_{\theta \in \Theta} \lambda^*_{\theta} r^{\gamma}_{\theta}(q)>0$.
\end{lemma}

\begin{proof}[ of Lemma \ref{lem:pos_rate}]  Pick $q \in \mathbf{S}$ and $\lambda^* \in C_{k^*}$. The derivative of $r_{\theta}^{\gamma}(q)$ with respect to $\gamma$, evaluated in $\gamma = 0$ is equal to
\[\sum_{x \in X} p^*(x) \left( \frac{p_{\theta}(x)}{\sum_{\theta' \in \Theta} q_{\theta'} p_{\theta'}(x)}-1\right)\]
Hence, the derivative of $\gamma \mapsto \sum_{\theta} \lambda^*_{\theta} r_{\theta}^{\gamma}(q)$ evaluated in $\gamma = 0$ is equal to
\begin{equation*} 
\sum_{x \in X} p^*(x)  \left(\frac{\sum_{\theta \in \Theta} \lambda^*_{\theta} p_{\theta}(x)}{\sum_{\theta \in \Theta} q_{\theta} p_{\theta}(x)}-1 \right)  \geq  \sum_{x \in X} p^*(x) \log  \left( \frac{\sum_{\theta \in \Theta} \lambda^*_{\theta} p_{\theta}(x)}{ \sum_{\theta \in \Theta} q_{\theta} p_{\theta}(x)}\right) \geq  0,
\end{equation*}
where the inequality follows from the concavity of the $\log$. Moreover, by definition of $C_{k^*}$,  
%(as the component of zeroes, where the cross-entropy is maximized)
 the right-hand side holds in equality if, and only if, $q \in C_{k^*}$. In addition, the quantity $\sum_{\theta \in \Theta} \lambda^*_{\theta} r^{\gamma}_{\theta}(q)$ is equal to zero if $\gamma = 0$. Therefore, if $q \notin C_{k^*}$, there exists $\gamma(q) >0$ such that, for any $0<\gamma < \gamma(q)$,
\[\sum_{\theta \in \Theta} \lambda^*_{\theta} r_{\theta}^{\gamma}(q) >0.\]
This concludes the proof. 
\end{proof}

\begin{remark} \label{rk1} Note that there does not exist a uniform threshold $\gamma_0>0$ on $\mathbf{S} \setminus C_{k^*}$: in words, for any $\gamma >0$, there might exist $q \in \mathbf{S} \setminus C_{k^*}$ such that $\sum_{\theta \in \Theta} \lambda^*_{\theta} r_{\theta}^{\gamma}(q) < 0$. However, a direct consequence of this lemma is that, for any open set $V$ such that $C_{k^*} \subset V$, there exists $\gamma_0 >0$ such that, for any $\gamma < \gamma_0$ 
\[\inf_{q \in \mathbf{S} \setminus V} \sum_{\theta \in \Theta} \lambda^*_{\theta} r_{\theta}^{\gamma}(q) > 0.\]
\end{remark}

An invariant probability measure for  $(q^{\gamma}_{n})_n$ is a probability measure $\pi$ on $\mathbf{S}$ such that 
\begin{equation} \label{eq:invariant} \int_{\mathbf{S}} h(q) \pi (dq) = \int_{\mathbf{S}} P_{\gamma}h(q) \pi(dq), \, \text{for all } \; h \in  \mathcal{C} \left( \mathbf{S} \right). \end{equation}
The set of such invariant probability distributions is denoted $\mathrm{Inv}\left(P_{\gamma} \right)$.
%By compactness of $\mathbf{S}$,  Since both $\mathbf{S}_0 := \partial \left(\mathbf{S}\right )$ and $\mathbf{S}^*$ are absorbing,  any invariantnd{equation}probability measure $\pi$ can be written as 
%\begin{equation} \label{eq:decomp}
%\pi = \alpha \pi_0 + (1-\alpha) \pi_1, \text{  with } \;  \pi_0(\mathbf{S}_0) = 1 \text{ and } \; \pi_1(\mathbf{S}^*) = 1.
%\end{equation} 
It is a compact and convex subset of $\mathcal{P}(\mathbf{S})$. Its extreme points are the so-called  \emph{ergodic measures} of $P_{\gamma}$.

\noindent  Given $\pi \in \mathcal{P}(\mathbf{S})$ and $\theta \in \Theta$, define the quantity $r^{\gamma}_{\theta}(\pi) := \int_{\mathbf{S}} r^{\gamma}_{\theta}(q) \pi(dq) $.
It is the expected growth of the belief in $\theta$ if the current belief is distributed according to $\pi$. Proposition 1 in \cite{BenSch19} states the following: if $\pi$ is an ergodic measure (for $P_{\gamma}$) then
\[\lim_{n \rightarrow + \infty} \frac{1}{n} \sum_{m=0}^{n-1} \log F_{\theta}^{\gamma}(q^{\gamma}_{m},x_{m+1}) = r^{\gamma}_{\theta}(\pi) \; \text{ for } \, \pi-\text{almost every initial condition } \, q_0.\]
Moreover $r^{\gamma}_{\theta}(\pi) = 0$ for any $\theta \in Supp(\pi)$, where 
\[Supp(\pi) := \left\{\theta \in \Theta: \; \pi (\{q \in \mathbf{S}: q_{\theta}>0 \})=1 \right\}.\]
  The intuition behind this is that, under stationary regime, the  supported states have a null expected growth.

%\MF{In previous version, we said: ``First note that  the Markov chain $(q^{\gamma}(n))_n$ is \emph{weakly Feller}, in the sense that the mapping $h   \mapsto P_{\gamma} h$ leaves $ \mathcal{C} \left(\mathbf{S}\right)$ invariant. It is a direct consequence from the fact that the maps $q \mapsto F^{\gamma}(q,x)$ are continuous, for all $x \in X$. Consequently the set of weak* limit points of $(\Pi^{\gamma}_n)_n$ is a non-empty compact subset of $\mathcal{P}(\mathbf{S})$", and it was in another lemma.  I am not sure we use that anywhere, do it is better to remove it. Tell me what you think.}

\begin{lemma} \label{lem:rate_inv} Let $\pi^{\gamma}$ be a  weak* limit point of the sequence $(\Pi^{\gamma}_n)_n$. Then $\pi^{\gamma} \in \mathrm{Inv}(P_{\gamma})$, and 
\[\sum_{\theta}\lambda^*_{\theta} r^{\gamma}_{\theta}(\pi^{\gamma}) \leq 0, \; \text{  for all } \;  \lambda^* \in C_{k^*}.\]
\end{lemma}

\begin{proof}[ of Lemma \ref{lem:rate_inv}]   Let $\pi^{\gamma}$ be a weak* limit point of the sequence $(\Pi_n^{\gamma})_n$.  By definition, there exists an increasing sequence $n_k$ such that,  for any $h \in \mathcal{C}\left(\mathbf{S} \right)$,
\[ \int_{\mathbf{S}} h(q) \pi^{\gamma}(dq)  = \lim_k \int_{\mathbf{S}} h(q) \Pi^{\gamma}_{n_k}(dq) = \lim_k \frac{1}{n_k} \sum_{m=0}^{n_k-1} h(q^{\gamma}_{m})\]
and
\[ \int_{\mathbf{S}} P_{\gamma} h(q) \pi^{\gamma}(dq)  =  \lim_k \frac{1}{n_k} \sum_{m=0}^{n_k-1} P_{\gamma} h(q^{\gamma}_{m})\]

\noindent By Lemma 3(i) in \cite{BenSch19}, for any $g \in \mathcal{C}(\mathbf{S} \times X)$, denoting $\overline{g}(q) := \sum_{x} p^*(x) g(q,x)$,    we have
\begin{equation} \label{eq:lemma3}
\lim_{n \rightarrow + \infty} \left| \frac{ \sum_{m=0}^{n-1} g(q^{\gamma}_{m},x_{m+1})}{n} -  \frac{\sum_{m=0}^{n-1} \overline{g}(q^{\gamma}_{m})}{n} \right| = 0, \; \text{ almost surely}. 
\end{equation}

\noindent Let $g(q,x) := h \left( q \circ F_{\gamma}(q,x)\right)$.\footnote{We let $q \circ q'$ denote the Hadamard product:$(q \circ q')_{\theta} := q_{\theta} q'_{\theta}$} Note that  $g(q^{\gamma}_{m}, x_{m+1}) = h\left(q^{\gamma}_{m} \circ F_{\gamma}\left(q^{\gamma}_{m},x_{m+1} \right)\right) = h\left(q^{\gamma}_{m+1}\right)$.
 Since $P_{\gamma} h(q) = \sum_{x \in X} p^*(x) h(q \circ F_{\gamma}(q,x)) = \bar{g}(q)$, and using (\ref{eq:lemma3}), we get 
\[\lim_{n \rightarrow + \infty} \left| \frac{ \sum_{m=0}^{n-1} h(q^{\gamma}_{m+1})}{n} -  \frac{\sum_{m=0}^{n-1} P_{\gamma} h(q^{\gamma}_{m})}{n} \right| = 0. \]
Hence
\[ \int_{\mathbf{S}} P_{\gamma} h(q) \pi^{\gamma}(dq)  =  \lim_k \frac{1}{n_k} \sum_{m=0}^{n_k-1} P_{\gamma} h(q^{\gamma}_{m}) =  \lim_k \frac{1}{n_k} \sum_{m=0}^{n_k-1} h(q^{\gamma}_{m}) = \int_{\mathbf{S}} h(q) \pi^{\gamma}(dq), \]
and (\ref{eq:invariant}) follows.

%Now, let $(n_k)_k$ be an increasing sequence such that 
%\[\lim_{k \rightarrow + \infty} \frac{1}{n_k} \sum_{m=0}^{n_k -1} \delta_{q^{\gamma}_{m}} = \pi^{\gamma}, \; \text{ for weak* convergence}.\] 
We now prove the second claim. Using identity (\ref{eq:lemma3}) with $g(q,x) = \sum_{\theta \in \Theta}\lambda^*_{\theta} \log F_{\theta}^{\gamma}(q,x)$ and using the fact that 
\[\sum_{m=0}^{n-1} \log F_{\theta}^{\gamma}\left(q^{\gamma}_{m},x_{m+1}\right) =
\sum_{m=0}^{n-1}  \log \frac{q_{m+1,\theta}^{\gamma}}{q_{m, \theta}^{\gamma}} = \log \frac{q_{n, \theta}^{\gamma}}{q_{0,\theta}^{\gamma}},\]
 we obtain
{\small
\[\lim_{n \rightarrow + \infty} \frac{\sum_{\theta \in \Theta} \lambda^*_{\theta} \left(\log q_{n,\theta}^{\gamma} - \log q_{0,\theta}^{\gamma} \right) - \sum_{m=0}^{n-1} \sum_{\theta \in \Theta} \lambda^*_{\theta} \sum_{x \in X} p^*(x) \log \left(F_{\theta}^{\gamma}\left(q^{\gamma}_{m},x \right)\right)}{n} = 0.\]
}
Note that $\limsup_n  \frac{\sum_{\theta \in \Theta} \lambda^*_{\theta} \log q_{n,\theta}^{\gamma}}{n}  \leq 0$ and $\lim_{n \rightarrow + \infty} \frac{\sum_{\theta \in \Theta} \lambda^*_{\theta} \log q_{0, \theta}^{\gamma}}{n}=0$. Moreover
\[\lim_{k \rightarrow + \infty} \frac{1}{n_k}  \sum_{m=0}^{n_k -1}  \sum_{x \in X} p^*(x) \log F_{\theta}^{\gamma} \left(q^{\gamma}_{m},x \right) = \int_{\mathbf{S}} \sum_{x \in X} p^*(x) \log F_{\theta}^{\gamma}(q,x)\pi(dq) = r^{\gamma}_{\theta}(\pi^{\gamma}).\] 
As a result,  we obtain that $\sum_{\theta \in \Theta} \lambda^*_{\theta}  r^{\gamma}_{\theta}(\pi^{\gamma}) \leq 0$.
\end{proof}

\begin{remark}
It is worth noting that, even though $\pi \in \mathcal{P}(\mathbf{S})$ is such that $\pi(\mathbf{S} \setminus C_{k^*})>0$, it does not imply that there exists $\gamma_0 >0$ such that, for all $\gamma < \gamma_0$, 
\[\sum_{\theta \in \Theta} \lambda^*_{\theta} r_{\theta}^{\gamma}(\pi) > 0.\] 
Indeed, regardless how small $\gamma$ is, $Supp (\pi)$ might contain points $q$ such that $\sum_{\theta \in \Theta} \lambda^*_{\theta} r_{\theta}^{\gamma}(q) <0$.

\noindent Importantly,  Lemmas \ref{lem:pos_rate} and  \ref{lem:rate_inv} do not allow us to  conclude  that there exists $\gamma_0 >0$ such that, for any $0< \gamma < \gamma_0$, we have $\pi^{\gamma}(C_{k^*}) = 1$.  As a matter of fact, this statement cannot hold, because an element of $\mathrm{Inv}(P_{\gamma})$ cannot be supported by $C_{k^*}$. To see this, suppose for the sake of simplicity that  $C_{k^*} = \{\lambda^*\}$. Then the distribution $\delta_{\{\lambda^*\}}$ is not invariant for $P_{\gamma}$.
\end{remark}

\begin{definition}\label{def:limiting measure}
The measure $\pi^* \in \Delta(\mathbf{S})$ is a limiting measure for the updating rule (\ref{C-Bay}) if there exists a sequence $(\gamma_{\ell},\pi_{\ell})_{\ell \in \mathbb{N}}$ such that:
\begin{enumerate}
\item $ \gamma_{\ell} \downarrow 0$,
\item $\pi_{\ell}$ is a weak* limit point of  $(\Pi_n^{\gamma_{\ell}})_n$, for all $\ell$,
\item $\lim_{\ell} \pi_{\ell} = \pi^*$ for the weak* topology. 
\end{enumerate}
\end{definition}

In words, the measure $\pi^*$ is a limiting measure if we can find a sequence of updating weights $(\gamma_{\ell})_{\ell}$ and a corresponding sequence $(\pi_{\ell})_{\ell}$ of limit points of $(\Pi_n^{\gamma_{\ell}})_n$, which converges to $\pi^*$ for the weak* topology.  We then prove the following:

\textbf{Theorem \ref{th:constant-step}.}
Let $\pi^*$ be a limiting measure for the updating rule (\ref{C-Bay}). The support of $\pi^*$ is included in $C_{k^*}$.

\begin{proof} [ of Theorem \ref{th:constant-step}]  Let $\pi^*$ be  a limiting measure and $(\gamma_{\ell},\pi_{\ell})_{\ell \in \mathbb{N}}$ the corresponding sequence of weights and invariant measures mentionned in Definition \ref{def:limiting measure}. %We want to prove that the support of $\pi^*$ is contained in $C_{k^*}$. \medskip 

Pick any $\lambda^* \in C_{k^*}$.  Given $\gamma >0$, $q \in \mathbf{S}$, define
\[R^*(\gamma,q) := \sum_{\theta \in \Theta} \lambda^*_{\theta} r_{\theta}^{\gamma}(q).\]
Note that $R^*(0,q) = 0$ and, by first point of Lemma \ref{lem:pos_rate}, 
\[R^*(\gamma,q) = R^*(0,q) + \gamma \frac{\partial R^*}{\partial \gamma}(0,q) + o(\gamma),  \text{ for all } \;  q \in \mathbf{S}.\]
%By Lemma \ref{lem:invariance_flow}, the support of $\pi^*$ is contained in the Birkhoff center (i.e. the closure of recurrent points) of the flow: 
%\[BC(\phi) := \{q \in \mathbf{S}: \exists t_k \uparrow + \infty \, \text{ such that } \; \lim_{k \rightarrow + \infty} \phi_{t_k}(q) = q \}.\]
%In this case (gradient-like vector field), we have $BC(\phi) = E$. Hence 
%\[Supp (\pi^*) \subseteq \bigcup_{k=1}^K C_k.\]

\noindent Let $O^*$ be an open subset of $\mathbf{S}$ such that 
\[\bigcup_{k \neq k^*} C_k \subseteq O, \;  C_{k^*} \subseteq O^*, \]
where $O:= \mathbf{S} \setminus \mathrm{Cl}(O^*)$. (Recall that the components $(C_k)_{k=1,\dots,k^*}$ are compact, connected and disjoint.)

Fix $\varepsilon >0$. Choose $O^*$ small enough so that $\left|\sup_{q \in \mathrm{Cl(O^*)}}  \frac{\partial R^*}{\partial \gamma}(0,q) \right| \leq \frac{\varepsilon}{2}$. (Recall that the derivative of $R^*$ is zero for all $q \in C_{k^*}$ and the derivative is continuous. Hence, such a choice is always possible.) This implies that $|R^*(\gamma,q)| \leq \epsilon \gamma$ for all $q \in \mathrm{Cl}(O^*)$ and  small enough $\gamma$. 
\vspace{.1cm}

\noindent Now define $c:= \frac{1}{2} \inf_{q \in O} \frac{\partial R^*}{\partial \gamma}(0,q) >0$.  We then have $R^*(\gamma,q) > c\gamma $ for all $q \in O$ and  small enough $\gamma$.

\noindent 
% Since $\pi^*(W) = 0$ and $W$ is closed, 
%\[\limsup_{p \rightarrow + \infty} \pi_p(W) \leq \pi^*(W) = 0.\] 
Finally, for large enough $\ell$,
\begin{eqnarray*}
R^*(\gamma_{\ell},\pi_{\ell}) &=& \int_{q \in \mathrm{Cl}(O^*)} R^*(\gamma_{\ell},q)\pi_{\ell}(dq) +  \int_{q \in O} R^*(\gamma_{\ell},q)\pi_{\ell}(dq)  \\
& > & - \varepsilon \gamma_{\ell} \pi_{\ell}(\mathrm{Cl}(O^*))  + c \gamma_{\ell}  \pi_{\ell}\left(O\right) 
\end{eqnarray*}
By Lemma \ref{lem:rate_inv}, $R^*(\gamma_{\ell},\pi_{\ell}) \leq 0$, meaning that $\pi_{\ell}(O) < \frac{\varepsilon}{c} \pi_{\ell}(\mathrm{Cl}(O^*)) \leq \frac{\varepsilon}{c}$ for large enough $\ell$. Since $O$ is open,
\[\pi^*(O) \leq \liminf_{\ell \rightarrow + \infty} \pi_{\ell}(O) < \frac{\varepsilon}{c}.\]
Since this inequality holds for any $\varepsilon >0$, we have $\pi^*(O) = 0$. In particular, this proves that
$\pi^*\left(N^{\delta}(C_{k^*}) \right) = 1$, where $N^{\delta}(C_{k^*})$ is the $\delta$-neighborhood of the set $C_{k^*}$.
% := \left\{q \in \mathbf{S}: \; \,  d\left(q,C_{k^*} \right) < \delta\right\}$. 
This concludes the proof.
\end{proof}

\begin{remark}  For any $\delta >0$, we have $\liminf_{\ell \rightarrow + \infty}\pi_{\ell}(N^{\delta} \left(C_{k^*}) \right) = 1$.
\end{remark}

\bibliographystyle{ecta}
\bibliography{biblio-non-bayesian-learning}

\end{document}